\documentclass[10pt, letterpaper]{article}
\usepackage[colorinlistoftodos]{todonotes}
\usepackage{epsfig}
\usepackage{lscape}
\usepackage{longtable}
\usepackage{amsmath}
\usepackage{amssymb}
\usepackage{amsthm}
\usepackage{setspace}
\usepackage{graphicx}
\usepackage{natbib}
\usepackage{color}
\usepackage{lineno}
\usepackage{subcaption}
\usepackage{framed}
\usepackage{tikz}
\usepackage[labelfont=bf,labelsep=period,justification=raggedright]{caption}
\usepackage{indent first}
\usepackage{hyperref}
\usepackage[utf8]{inputenc}

\newtheorem{theo}{Theorem}[]
\newtheorem{coro}[theo]{Corollary}
\newtheorem{lemma}[theo]{Lemma}

\newtheorem{rem}[theo]{Remark}
\newtheorem{defi}[theo]{Definition}
\newtheorem{prop}[theo]{Proposition}

\title{{Roadblocked monotonic paths and the enumeration of coalescent histories for non-matching caterpillar gene trees and species trees}}

\author{Zoe M.~Himwich\thanks{Department of Mathematics, Stanford University, Stanford, CA 94305 USA} \\ Noah A.~Rosenberg\thanks{Department of Biology, Stanford University, Stanford, CA 94305 USA. Email: noahr@stanford.edu.}}

\date{
{\normalsize  \today} \medskip
{\bf Running title:} Coalescent histories for non-matching caterpillars 
}

\begin{document}
\maketitle
\normalsize

\noindent {\bf Abstract.} Given a gene tree topology and a species tree topology, a coalescent history represents a possible mapping of the list of gene tree coalescences to associated branches of a species tree on which those coalescences take place. Enumerative properties of coalescent histories have been of interest in the analysis of relationships between gene trees and species trees. The simplest  enumerative result identifies a bijection between coalescent histories for a matching caterpillar gene tree and species tree with monotonic paths that do not cross the diagonal of a square lattice, establishing that the associated number of coalescent histories for $n$-taxon matching caterpillar trees ($n \geqslant 2$) is the Catalan number $C_{n-1} = \frac{1}{n} {2n-2 \choose n-1}$. Here, we show that a similar bijection applies for \emph{non-matching} caterpillars, connecting coalescent histories for a non-matching caterpillar gene tree and species tree to a class of \emph{roadblocked} monotonic paths. The result provides a simplified algorithm for enumerating coalescent histories in the non-matching caterpillar case. It enables a rapid proof of a known result that given a caterpillar species tree, no non-matching caterpillar gene tree has a number of coalescent histories exceeding that of the matching gene tree. Additional results on coalescent histories can be obtained by a bijection between permissible roadblocked monotonic paths and Dyck paths. We study the number of coalescent histories for non-matching caterpillar gene trees that differ from the species tree by nearest-neighbor-interchange and subtree-prune-and-regraft moves, characterizing the non-matching caterpillar with the largest number of coalescent histories. We discuss the implications of the results for the study of the combinatorics of gene trees and species trees.

%%%%%%%%%%%%%%%%%%%%%%%%%%%%%%%%%%%%%%%%%%%%%%%%%%%%%%%%%%%%%%%%%%%%%%%
\section{Introduction}
\label{secIntroduction}

\noindent In the mathematical study of evolutionary trees, genetic lineages can be treated as evolving along the branches of a species phylogeny, a tree that represents the evolutionary relationships among a set of species \citep{PamiloAndNei88,Maddison97,DegnanAndRosenberg09}. A tree describing a set of genetic lineages that descend from a common ancestor is a \emph{gene tree}, and a tree relating the species themselves is a \emph{species tree}. Looking backward in time, in a gene tree of genetic lineages sampled from representative individuals of a given set of species, a pair of genetic lineages can \emph{coalesce}, or find a common ancestor, only after the common ancestor of their species is reached. More generally, a set of two or more genetic lineages has a most recent common ancestor only after the most recent common ancestor of their associated species is reached.

The study of the relationship between gene trees and species trees---usually treated as binary, rooted, and leaf-labeled---has generated a number of novel combinatorial structures \citep{Maddison97, DegnanAndSalter05, RosenbergAndTao08, ThanAndNakhleh09, DegnanEtAl12:mathbiosci, StadlerAndDegnan12, Wu12, Wu16, DegnanAndRhodes15}. Among these are \emph{coalescent histories}, structures that describe the possible locations on a species tree where the coalescences of a gene tree can take place \citep{DegnanAndSalter05, Rosenberg07:jcb}. More precisely, for a (binary, rooted, leaf-labeled) gene tree topology $G$ and a (binary, rooted, leaf-labeled) species tree topology $S$ on the same set of taxa, a coalescent history $f$ associates with each coalescence in $G$ an edge of $S$, such that two properties are satisfied: (i) the species tree edge $h(u)$ associated with a gene tree coalescence $u$ is ancestral to all lineages that descend from $u$; (ii) for any pair of gene tree coalescences $u,v$ for which $u$ lies on a path from $v$ to a leaf of the gene tree, $h(u)$ lies on a path from $h(v)$ to a leaf of the species tree. From a biological perspective, this pair of constraints encodes the rules that (i) gene lineages can coalesce only in a branch of the species tree in which it is possible for their ancestors to coexist, and that (ii) ancestors can coalesce no more recently than their descendants.

\cite{Rosenberg07:jcb} provided a recursion that enumerates coalescent histories for arbitrary gene tree and species tree topologies. For gene tree topology $G$ and species tree topology $S$, so that the taxon set of $S$ is a superset of that of $G$ but not necessarily the same set, let $T(G,S)$ denote the minimal displayed subtree of $S$ that contains all the taxa of $G$, that is, the subtree of $S$ rooted at the node that corresponds to the most recent common ancestor of the taxa with the same labels as the taxa in $G$. Let $d(G,S) \geqslant 0$ denote the number of edges that separate the root of $T(G,S)$ from the root of $S$. Let $G_L$ and $G_R$ denote the left and right subtrees of $G$. We define an integer parameter $m \geqslant 1$, and write a recursion for a function $B_{G,S,m}$:
\begin{equation}
\label{eqRecursion}
B_{G,S,m} = \sum_{k=1}^m B_{G_L, T(G_L,S),k+d(G_L,S)} B_{G_R, T(G_R,S),k+d(G_R,S)},
\end{equation}
The base case is obtained by setting $B_{G,S,m}$ to 1 for all $m$ in the case that $G$ has only one taxon. With these definitions, the number of coalescent histories for gene tree topology $G$ and species tree topology $S$ is $B_{G,S,1}$.

Caterpillar species trees, in which an internal node exists that is descended from all other internal nodes, represent a special case in which enumeration of the coalescent histories is simpler than in the general case of arbitrary species trees. Thus, although exact and asymptotic results are known for certain additional shapes \citep{Rosenberg07:jcb, Rosenberg19,  DisantoAndRosenberg15}, enumerative properties have been explored most extensively for caterpillar species trees and shapes that closely resemble them \citep{Degnan05, DegnanAndSalter05, Rosenberg07:jcb, Rosenberg13:tcbb, RosenbergAndDegnan10,  DisantoAndRosenberg16}. First, for a matching caterpillar gene tree and species tree---a caterpillar gene tree and species tree with the same labeled topology---\cite{Degnan05} found a bijection between coalescent histories and monotonic paths on a square lattice that do not cross above the $y=x$ diagonal, a quantity well-known to be described by the Catalan number sequence \citep[][item 24]{Stanley15}. Eq.~\ref{eqRecursion} recovers the Catalan numbers in this case \citep[][Corollary 3.5]{Rosenberg07:jcb}, and can be used to show that the number of coalescent histories for matching gene trees and species trees in small ``caterpillar-like families'' is asymptotic to a constant multiple of the Catalan numbers \citep{Rosenberg07:jcb, Rosenberg13:tcbb}. This asymptotic behavior has been demonstrated for caterpillar-like families of arbitrary size using techniques of analytic combinatorics \citep{DisantoAndRosenberg16}.

Enumerative results have been comparatively little studied, however, in the case that labeled gene trees and species trees disagree in topology. \cite{ThanEtAl07} performed a numerical investigation, finding that the number of coalescent histories for non-matching gene tree and species tree topologies generally decreases with increasing subtree-prune-and-regraft (SPR) distance between the trees. \cite{RosenbergAndDegnan10} demonstrated that for the caterpillar species tree topology with $n \geqslant 7$ taxa, there exists a non-matching gene tree topology with more coalescent histories than the matching caterpillar gene tree topology. Nevertheless, for caterpillar species tree topologies, \cite{DegnanAndRhodes15} showed that no non-matching caterpillar gene tree topology can exceed the matching caterpillar gene tree topology in number of coalescent histories; indeed, the constructive example of \cite{RosenbergAndDegnan10} of a non-matching gene tree topology with more coalescent histories than the matching caterpillar was not itself a caterpillar.

Here, we extend the monotonic path approach of \cite{Degnan05} to non-matching caterpillar gene tree and species tree topologies. We show that coalescent histories for non-matching caterpillar gene tree and species tree topologies can be bijectively associated with a set of \emph{roadblocked} monotonic paths that do not cross above the $y=x$ diagonal of a square lattice. The approach immediately recovers the result of \cite{DegnanAndRhodes15} that non-matching caterpillar gene tree topologies do not exceed the matching caterpillar gene tree topology in number of coalescent histories. It enables calculations of the number of coalescent histories for caterpillar gene tree topologies that differ from the species tree by common transformations---nearest-neighbor-interchange and subtree-prune-and-regraft. We characterize non-matching caterpillar gene trees with the largest numbers of coalescent histories, finding that the number of coalescent histories in such cases is asymptotically equivalent to that in the matching case.

%%%%%%%%%%%%%%%%%%%%%%%%%%%%%%%%%%%%%%%%%%%%%%%%%%%%%%%%%%%%%%%%%%%%%%%
\section{Preliminaries}
\label{secPreliminaries}

%%%%%%%%%%%%%%%%%%%%%%%%%%%%%%%%%%%%%%%%%%%%%%%%%%%%%%%%%%%%%%%%%%%%%%%
\subsection{Caterpillar trees}
\label{secCaterpillarTrees}

\noindent We consider binary, rooted, leaf-labeled trees with leaf labels bijectively drawn from a label set $X$ containing $n$ distinct labels. For convenience, a ``tree'' refers to a binary, rooted, leaf-labeled tree. Trees contain two types of nodes, leaf nodes and non-leaf, or internal, nodes. Because trees are rooted, we say that a node $v_1$ of a tree $G$ is \emph{descended} from another node $v_2$ if the shortest path from $v_1$ to the root node contains $v_2$. We also say that $v_2$ is \emph{ancestral} to $v_1$. Ancestor--descendant relationships also apply to pairs of edges and to pairs containing a vertex and an edge. A node or edge is trivially descended from itself, and it is also trivially ancestral to itself. The root node is an internal node.

We focus on \emph{caterpillar trees}, trees in which there exists an internal node descended from all other internal nodes (Figure~\ref{figTransformations}A). A caterpillar tree has exactly one \emph{cherry} node, a node with exactly two descendant leaves. Among leaves, the longest path length to the root of a caterpillar tree with $n$ leaves is $n-1$.

The number of distinct caterpillar trees possible for a label set $X$ with $n$ distinct labels is $n!/2$: the leaf separated from the root by only edge has $n$ possible labels, the leaf two edges from the root then has $n-1$ possible labels, and so on. In this assignment of labels, the leaves descended from the cherry node are exchangeable. Hence, only one labeling is possible for these leaves, giving a total of $n(n-1)(n-2)\times \cdots \times 3 = n!/2$ labelings. These labelings represent the $n!/2$ caterpillar \emph{labeled topologies} for label set $X$.

For convenience, we organize the labels in an $n$-leaf caterpillar tree $G$ canonically in a vector $\mathbf{g}$ of length $n$. For $i=3,4,\ldots,n$, entry $i$ in the vector is the label of the leaf separated from the root by $n-i+1$ edges. Entries 1 and 2 are the labels for the leaves in the cherry. Two vectors of labels $\mathbf{g}$ and $\mathbf{s}$ are considered to be equivalent if and only if one of the following two conditions holds: (1) $g_i=s_i$ for all $i$, or (2) $g_1=s_2$, $g_2=s_1$, and $g_i=s_i$ for each $i=3,4,\ldots,n$.

Two leaves in a caterpillar tree are considered to be \emph{adjacent} if they are separated by exactly two or three edges (Figure~\ref{figTransformations}A). Equivalently, leaves are adjacent if and only if their indices in the sequence of labels for the tree differ by 1, or if one is entry 1 and the other is entry 3.

A \emph{component} of a caterpillar tree is a subset of adjacent leaves, excluding from the definition the subset consisting solely of the pair of leaves in the cherry. Formally, a subset of labels $X^\prime \subset X$ is a component of $G$ if $X^\prime \neq \{x_1, x_2\}$ and for any pair of labels $x_1, x_2 \in X^\prime$, there exists a sequence of distinct elements $x_1, x_{i_1},x_{i_2}, \ldots x_{i_j}, x_2 \in X^\prime$ in which each consecutive pair of elements labels adjacent leaves in $G$.

It is convenient to number the internal nodes of an $n$-leaf caterpillar tree from 1 to $n-1$ in increasing order from the cherry node toward the root. These nodes are ordered by ancestor-descendant relationships, so that the node of smallest value in any nonempty subset of internal nodes descends from all other elements of the subset. We call this node the \emph{minimal} node of the subset. It is also useful to consider that a tree possesses an internal edge ancestral to its root node; thus, identifying each internal node with its immediate ancestral edge, a nonempty subset of internal edges has a minimal edge.

%%%%%%%%%%%%%%%%%%%%%%%%%%%%%%%%%%%%%%%%%%%%%%%%%%%%%%%%%%%%%%%%%%%%%%%
\subsection{Relationships between pairs of caterpillar trees}
\label{secRelationships}

\noindent The labelings of distinct caterpillar trees with the same label set differ by a permutation of the vector of leaf labels. We will have occasion to examine pairs of caterpillar trees whose labelings differ by specific types of permutation: \emph{nearest-neighbor-interchange} and \emph{subtree-prune-and-regraft} \citep{Steel16}.

Consider two distinct caterpillar trees $G$ and $S$, bijectively labeled from the same set of $n$ distinct labels.

\begin{defi}
Caterpillar trees $G$ and $S$ differ by a \emph{nearest-neighbor-interchange}, or \emph{NNI} move, if $S$ can be obtained from $G$ by exchanging the labels of a pair of adjacent leaves in $G$ that are separated by exactly three edges (Figure~\ref{figTransformations}B).
\label{defiNNI}
\end{defi}

\noindent Note that our definition of adjacent leaves includes the leaves corresponding to labels $g_1$ and $g_2$ in the canonical ordering. This pair is the only pair of adjacent leaves that are not separated by an NNI move.

\begin{defi}
Caterpillar trees $G$ and $S$ differ by a \emph{subtree-prune-and-regraft}, or \emph{SPR} move, if there exists an ordered pair of edges $(e_1,e_2)$ in $G$ with the property that if edge $e_1$ is cut, edge $e_2$ is subdivided in two by placement of a new vertex $v$ of degree two, and the subtree descended from $e_1$ is connected to vertex $v$ such that $v$ now has degree three and is ancestral to the subtree, then tree $S$ is obtained (Figure~\ref{figTransformations}C, \ref{figTransformations}D).
\label{defiSPR}
\end{defi}

\noindent In an SPR move, note that it is possible for the edge $e_2$ to be the edge ancestral to the root of $G$.

\begin{defi}
Caterpillar trees $G$ and $S$ differ by a \emph{cyclic permutation} if there exists a component $G^\prime$ of $G$ and a component $S^\prime$ of $S$ such that the labels of $S^\prime$ represent a cyclic permutation of the labels of $G^\prime$.
\label{defiCyclic}
\end{defi}

\noindent By definition of a component, this definition excludes permutations that simultaneously involve leaves separated from the root by the fewest edges and leaves separated from the root by the most edges, unless all leaves are involved.

\begin{defi}
Caterpillar trees $G$ and $S$ differ by an \emph{incrementation} if they differ by a cyclic permutation and at most one label has positions in the canonical label vectors of $G$ and $S$ that differ by more than one.
\label{defiIncrementation}
\end{defi}

\noindent $S$ can differ from $G$ by a \emph{forward} or a \emph{reverse} cycle or incrementation (Figure~\ref{figTransformations}C, \ref{figTransformations}D). If $S$ differs from $G$ by a forward incrementation or cycle, then $G$ differs from $S$ by a reverse incrementation or cycle, and vice versa. Note that each cyclic permutation that exchanges two leaves is concurrently a forward incrementation, a reverse incrementation, and an NNI move.

We can immediately observe that a pair of caterpillar trees $G$ and $S$ differ by an SPR move if and only if they also differ by an incrementation of the leaf labels. SPR moves that convert caterpillars to caterpillars necessarily prune and regraft a single leaf. If a leaf is pruned from $G$ and regrafted to $S$, then depending on which leaf is pruned and where it is regrafted, $S$ can differ from $G$ by either a forward or a reverse incrementation.  Therefore, enumeration of coalescent histories in the case that caterpillar trees differ by an SPR move is performed by enumeration in the associated case of a forward or a reverse incrementation.

%%%%%%%%%%%%%%%%%%%%%%%%%%%%%%%%%%%%%%%%%%%%%%%%%%%%%%%%%%%%%%%%%%
%%%%%%%%%%%%%%%%%%%%%%%%%%%%% Figure 1 %%%%%%%%%%%%%%%%%%%%%%%%%%%
%%%%%%%%%%%%%%%%%%%%%%%%%%%%%%%%%%%%%%%%%%%%%%%%%%%%%%%%%%%%%%%%%%
\begin{figure}[tpb]
\vskip -2cm
\begin{center}
\includegraphics[width=9in]{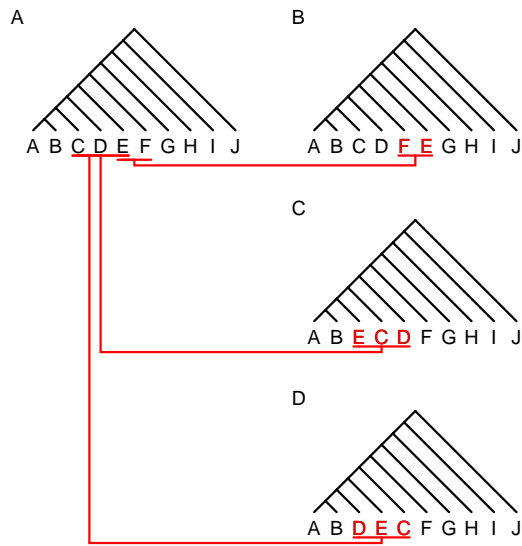}
\end{center}
\vskip -20.5cm
\caption{Transformations of caterpillar trees. (A) A caterpillar tree $G_1$. The vector of labels for $G_1$, in canonical order, is $(A,B,C,D,E,F,G,H,I,J)$. The adjacent pairs of leaves are $(A,B)$, $(A,C)$, $(B,C)$, $(C,D)$, $(D,E)$, $(E,F)$, $(F,G)$, $(G,H)$, $(H,I)$, and $(I,J)$. (B) A tree $G_2$ that differs from $G_1$ by nearest-neighbor-interchange. Leaves $E$ and $F$ are exchanged. (C) A tree obtained from $G_1$ by forward incrementation of leaves $C$, $D$, and $E$. (D) A tree obtained from $G_1$ by reverse incrementation of leaves $C$, $D$, and $E$. The tree in (C) can also be viewed as the result of a subtree-prune-and-regraft operation, with the branch leading to leaf $E$ pruned and regrafted; the tree in (D) can be viewed as the result of an SPR operation involving  the leaf leading to $C$. In each panel, the red line indicates which leaves are permuted.}
\label{figTransformations}
\end{figure}
%%%%%%%%%%%%%%%%%%%%%%%%%%%%%%%%%%%%%%%%%%%%%%%%%%%%%%%%%%%%%%%%%%%%%%%
%%%%%%%%%%%%%%%%%%%%%%%%%%%%%%%%%%%%%%%%%%%%%%%%%%%%%%%%%%%%%%%%%%%%%%%
%%%%%%%%%%%%%%%%%%%%%%%%%%%%%%%%%%%%%%%%%%%%%%%%%%%%%%%%%%%%%%%%%%%%%%%

%%%%%%%%%%%%%%%%%%%%%%%%%%%%%%%%%%%%%%%%%%%%%%%%%%%%%%%%%%%%%%%%%%%%%%%
\subsection{Coalescent histories}
\label{secCoalescentHistories}

\noindent We study coalescent histories for a caterpillar \emph{gene tree} $G$ and a caterpillar \emph{species tree} $S$, treated as binary, rooted, leaf-labeled caterpillar trees, each with $n$ leaves labeled by labels bijectively drawn from the same set $X$. This setting corresponds to considering $G$ to represent the tree formed by sampling a single gene lineage in each of the $n$ species present in species tree $S$. Gene tree $G$ and species tree $S$ are said to be \emph{matching} if $G$ and $S$ have the same labeled topology, and they are said to be \emph{non-matching} otherwise.

Formally, a coalescent history can be defined as follows \citep{RosenbergAndDegnan10}.

\begin{defi}
\label{defiCoalhist}
Consider an ordered pair of binary, rooted, leaf-labeled trees $(G,S)$ whose labels are  bijectively drawn from the same label set $X$. A \emph{coalescent history} is a function $h$ from the set of internal nodes of $G$ to the set of internal edges of $S$ that satisfies two conditions:
\begin{enumerate}
	\item For each internal node $v$ of $G$, all leaf labels for leaves descended from $v$ in $G$ label leaves descended from edge $h(v)$ in $S$.
    \item For all pairs of internal nodes $v_1,v_2$ in $G$, if node $v_2$ is descended from node $v_1$ in $G$, then edge $h(v_2)$ is descended from edge $h(v_1)$ in $S$.
\end{enumerate}
\end{defi}

\noindent An illustration appears in Figure \ref{figCoalescentHistories}. Recall that we consider that $S$ contains an edge ancestral to its root; this edge can be the image of an internal node of $G$ under a coalescent history mapping. Note that because an edge is trivially descended from itself, in part 2 of Definition \ref{defiCoalhist}, it is permissible for $h(v_2)$ to equal $h(v_1)$.

We will have occasion to use the concept of a partial coalescent history.

\begin{defi}
\label{defiPartial}
Consider an ordered pair of binary, rooted, leaf-labeled trees $(G,S)$ whose labels are drawn from the same label set $X$, not necessarily bijectively. A \emph{partial coalescent history} is a function $h$ from the set of internal nodes of $G$ to the set of internal edges of $S$, satisfying the two conditions in Definition \ref{defiCoalhist}.
\end{defi}

\noindent We say that if $G$ is empty, then $(G,S)$ has one partial coalescent history. For nonempty $G$, because the labels in $G$ are not necessarily the same as those of $S$, it is possible that for some nodes $v$ in $G$, $S$ has no edge that can serve as the image of a node in $G$. In this case, the pair $(G,S)$ has no partial coalescent histories.
When connecting the purely graphical definition of coalescent histories in Definition \ref{defiCoalhist} to the biological context in which they arise, we say that an internal node $v$ of $G$ is a gene tree \emph{coalescence}; the coalescence is said to occur on edge $h(v)$ of $S$.

%%%%%%%%%%%%%%%%%%%%%%%%%%%%%%%%%%%%%%%%%%%%%%%%%%%%%%%%%%%%%%%%%%
%%%%%%%%%%%%%%%%%%%%%%%%%%%%% Figure 2 %%%%%%%%%%%%%%%%%%%%%%%%%%%
%%%%%%%%%%%%%%%%%%%%%%%%%%%%%%%%%%%%%%%%%%%%%%%%%%%%%%%%%%%%%%%%%%
\begin{figure}[tpb]
\vskip -1cm
\begin{center}
\includegraphics[width=\textwidth]{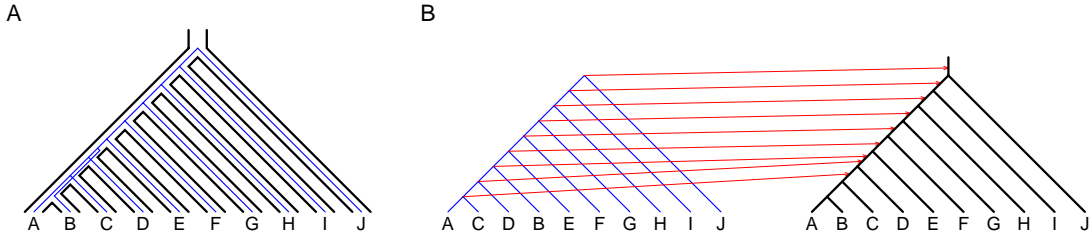}
\end{center}
\vskip -18.5cm
\caption{Coalescent histories. (A) A gene tree $G$ and species tree $S$ with the same label set. The gene tree appears in blue, and the species tree appears in black. (B) The coalescent history depicted in (A) for $(G,S)$. The arrows connect internal nodes of $G$ to their associated edges in $S$.}
\label{figCoalescentHistories}
\end{figure}
%%%%%%%%%%%%%%%%%%%%%%%%%%%%%%%%%%%%%%%%%%%%%%%%%%%%%%%%%%%%%%%%%%%%%%%
%%%%%%%%%%%%%%%%%%%%%%%%%%%%%%%%%%%%%%%%%%%%%%%%%%%%%%%%%%%%%%%%%%%%%%%
%%%%%%%%%%%%%%%%%%%%%%%%%%%%%%%%%%%%%%%%%%%%%%%%%%%%%%%%%%%%%%%%%%%%%%%

%%%%%%%%%%%%%%%%%%%%%%%%%%%%%%%%%%%%%%%%%%%%%%%%%%%%%%%%%%%%%%%%%%%%%%%
\subsection{Catalan numbers and monotonic paths}
\label{secCatalan}

\noindent We recall a number of results concerning Catalan numbers and their use in counting paths along the edges of square lattices. The \emph{Catalan sequence} $\{C_n\}_{n \geqslant 0}$ satisfies
\begin{equation*}
C_n = \frac{1}{n+1}{2n \choose n},
\end{equation*}
beginning from $n=0$, with values 1, 1, 2, 5, 14, 42, 132, 429, 1430, 4862, \ldots

Catalan numbers can be placed in the combinatorial construction known as \emph{Catalan's triangle} \citep{Reuveni14}, of which we display the first several columns:
\begin{equation*}
\begin{matrix}
  &   &   &   &    & 42 \\
  &   &   &   & 14 & 42 \\
  &   &   & 5 & 14 & 28 \\
  &   & 2 & 5 &  9 & 14 \\
  & 1 & 2 & 3 &  4 &  5 \\
1 & 1 & 1 & 1 &  1 &  1
\end{matrix}
\end{equation*}
In this triangle, the initial 1 in the lower left corner is denoted $D(0,0)$. Other entries are denoted $D(n,k)$, with $n$ as the horizontal distance from the lower left corner and $k$ as the vertical distance from this entry.

For $n,k$ with $0 \leqslant k \leqslant n$, the entries $(n,k)$ satisfy the recursion relation
\begin{equation}
\label{eqRecursionCatalan}
    D(n,k)=D(n,k-1)+D(n-1,k),
\end{equation}
with initial condition $D(0,0)=1$. The general formula for $D(n,k)$ is
\begin{equation}
    D(n,k)=
    \begin{cases}
    1 & k=0 \\
    {n+k \choose k}-{n+k \choose k-1} & 1\leqslant k\leqslant n \\
    0 & k>n.
    \end{cases}
\label{eqCatalansTriangle}
\end{equation}
In particular, for $k=n$, we have $D(n,n)=C_n$.

The entry $D(n,k)$ counts the number of monotonic paths on the lattce in the first quadrant of the $(n,k)$ plane (including the coordinate axes) that do not cross the line $k=n$, where a \emph{monotonic path} is a path from $(0,0)$ to $(n,k)$ that proceeds by steps upward and to the right on the lattice.

We will also make use of extensions of Catalan's triangle known as \emph{Catalan's trapezoids of order $m$}, which contain an initial column of $m$ entries equal to 1, rather than a single entry~\citep{Reuveni14}. Entries $D_m(n,k)$ in Catalan's trapezoids satisfy a version of eq.~\ref{eqRecursionCatalan}:
\begin{equation}
\label{eqRecursionCatalanTrapezoid}
    D_m(n,k)=D_m(n,k-1)+D_m(n-1,k).
\end{equation}
We have $D_1(n,k)=D(n,k)$. The first few columns of Catalan's trapezoid of order 3 appear below:
\begin{equation*}
\begin{matrix}
  &   &    &    & 90 \\
  &   &    & 28 & 90 \\
  &   & 9  & 28 & 62 \\
  & 3 & 9  & 19 & 34 \\
1 & 3 & 6  & 10 & 15 \\
1 & 2 & 3  & 4  & 5 \\
1 & 1 & 1  & 1  & 1 \\
\end{matrix}
\end{equation*}

An entry in the trapezoid can be calculated in closed form as
\begin{equation}
    D_{m}(n,k)=\begin{cases} {n+k \choose k} & 0\leqslant k <m \\ {n+k \choose k} - {n+k \choose k-m} & m\leqslant k\leqslant n+m-1 \\ 0 & k > n+m-1.
    \end{cases}
\label{eqCatalansTrapezoid}
\end{equation}
The entry $D_m(n,k)$ in Catalan's trapezoid of order $m$ counts the number of monotonic paths on the lattice in the first quadrant of the $(n,k)$ plane (including the coordinate axes) that do not cross the line $k=n+m-1$.

%%%%%%%%%%%%%%%%%%%%%%%%%%%%%%%%%%%%%%%%%%%%%%%%%%%%%%%%%%%%%%%%%%%%%%%
\section{Bijection of coalescent histories and roadblocked monotonic paths}
\label{secBijection}

%%%%%%%%%%%%%%%%%%%%%%%%%%%%%%%%%%%%%%%%%%%%%%%%%%%%%%%%%%%%%%%%%%%%%%%
\subsection{Matching gene trees and species trees}
\label{secMatching}

\noindent \cite{Degnan05} proved that the number of coalescent histories for a matching caterpillar gene tree $G$ and species tree $S$ with $n$ labels is the Catalan number $C_{n-1}$, demonstrating a bijection between coalescent histories and monotonic paths that do not cross the $y=x$ diagonal of a square lattice. We will discuss this well-known correspondence, as the bijective approach is useful for the non-matching case.

\begin{lemma}
The coalescent histories for a matching $n$-leaf caterpillar gene tree $G$ and species tree $S$ can be bijectively associated with monotonic paths that do not cross the $y=x$ diagonal of an $(n-1) \times (n-1)$ lattice.
\label{lemBijection}
\end{lemma}

\begin{proof}
Label the internal nodes of $G$ sequentially from 1 to $n-1$, using 1 for the internal node nearest the cherry and $n-1$ for the root. For each internal node of $G$, identify the label for the node with the edge immediately ancestral to it. Similarly, sequentially label the internal nodes of $S$ from 1 to $n-1$, proceeding from the cherry toward the root and identifying the label for each node with its immediate ancestral edge.

For each $j$ with $1 \leqslant j \leqslant n-1$, denote by $G_j$ the subtree of the gene tree rooted at node $j$, and for each $i$ with $1 \leqslant i \leqslant n-1$, denote by $S_i$ the subtree of the species tree rooted at node $i$. We also define $G_0$ and $S_0$ to be empty subtrees of the gene tree and species tree, respectively. Denote by $A_{i,j}$ the set of partial coalescent histories for $(G_j,S_i)$. For matching $G$ and $S$, for each $j$ with $0 \leqslant j \leqslant n-1$, $G_j=S_j$. Hence, by definition of a coalescent history, for each internal node $j\geqslant 1$ of $G$, the image $h(j)$ in a coalescent history $h$ of $(G,S)$ must be ancestral in $S$ to all leaves of $S$ labeled by labels in $G_j$. The edges of $S$ with this property are edges $j, j+1, \ldots, n-1$. For $j \geqslant 1$, we have $j \leqslant h(j) \leqslant n-1$, and $|A_{i,j}|=0$ for all $(i,j)$ with $i < j$.

Each partial coalescent history in $A_{i,j}$ is formed in one of two ways. Gene tree node $j \geqslant 1$ is mapped either to species tree internal edge $i$, or to one of the edges $1,2,\ldots,i-1$. The former case produces $|A_{i,j-1}|$ partial coalescent histories, each obtained by appending the coalescence of gene tree node $j$ to a partial coalescent history for $(G_{j-1},S_i)$. The latter case produces $|A_{i-1,j}|$ partial coalescent histories; because no gene tree coalescences in such a partial coalescent history occur on species tree edge $i$, each such partial coalescent history for $(G_j,S_i)$ is a partial coalescent history for $(G_j,S_{i-1})$. Hence, we have
\begin{equation}
\label{eqRecursionA}
|A_{i,j}| = |A_{i,j-1}| + |A_{i-1,j}|,
\end{equation}
with the constraint $|A_{i,j}|=0$ for $j\geqslant 1$ and $i < j$. For $j=0$ and $0 \leqslant i \leqslant n-1$, we have $|A_{i,0}|=1$ by the convention that $(G,S)$ has one partial coalescent history for empty $G$. We set $|A_{i,j}|=0$ for all $(i,j)$ that do not satisfy $0 \leqslant i,j \leqslant n-1$.

Recursion \ref{eqRecursionA} and its base cases, with $i$ in the role of $n$ and $j$ in the role of $k$, is precisely eq.~\ref{eqRecursionCatalan}. Setting $i=j=n-1$, eq.~\ref{eqRecursionCatalan} gives the recursion for enumerating the set of monotonic paths that do not cross the $y=x$ diagonal of an $(n-1) \times (n-1)$ square lattice, a set with $C_{n-1}$ elements. In the bijection between coalescent histories and monotonic paths, each step to the right in the lattice, incrementing $i$, corresponds to incorporating an additional edge of the species tree as a possible location for gene tree coalescences, and each step up, incrementing $j$, corresponds to occurrence of a gene tree coalescence.
\end{proof}

We can read a coalescent history of $(G,S)$ from its associated monotonic path (Figure \ref{figLatticeMatching}). For example, in a 10-leaf tree, the monotonic path that proceeds through (0,0), (3,0), (3,2), (6,2), (6,3), (7,3), (7,7), (9,7), and (9,9) has no gene tree coalescences on edge 1 of the species tree above $(A,B)$ or on edge 2 above $((A,B),C)$. Gene tree coalescences $(A,B)$ and $((A,B),C)$ occur on edge 3 above species tree node $(((A,B),C),D)$.  No gene tree coalescences occur on edges 4 or 5. Gene tree coalescence $(((A,B),C),D)$ occurs on edge 6. Four gene tree coalescences occur on edge 7 above species tree node $(((((((A,B),C),D),E),F),G),H)$. The two remaining gene tree coalescences occur on edge 9 above the species tree root.

The bijection between coalescent histories and monotonic paths generates a set of values of $|A_{i,j}|$ that considers each $i$ and $j$ with $0 \leqslant i,j \leqslant n-1$ and $i \geqslant j$. These values can be depicted in a lattice so that the value $|A_{i,j}|$ is associated with the coordinate of lattice point $(i,j)$ (Figure \ref{figLatticeMatching}). Indeed, they correspond exactly to the entries of Catalan's triangle (eq.~\ref{eqCatalansTriangle}), with $i$ in the role of $n$ and $j$ in the role of $k$.

%%%%%%%%%%%%%%%%%%%%%%%%%%%%%%%%%%%%%%%%%%%%%%%%%%%%%%%%%%%%%%%%%%
%%%%%%%%%%%%%%%%%%%%%%%%%%% Figure 3 %%%%%%%%%%%%%%%%%%%%%%%%%%%%%
%%%%%%%%%%%%%%%%%%%%%%%%%%%%%%%%%%%%%%%%%%%%%%%%%%%%%%%%%%%%%%%%%%
\begin{figure}[tpb]
\vskip -1.4cm
\begin{center}
\includegraphics[width=7in]{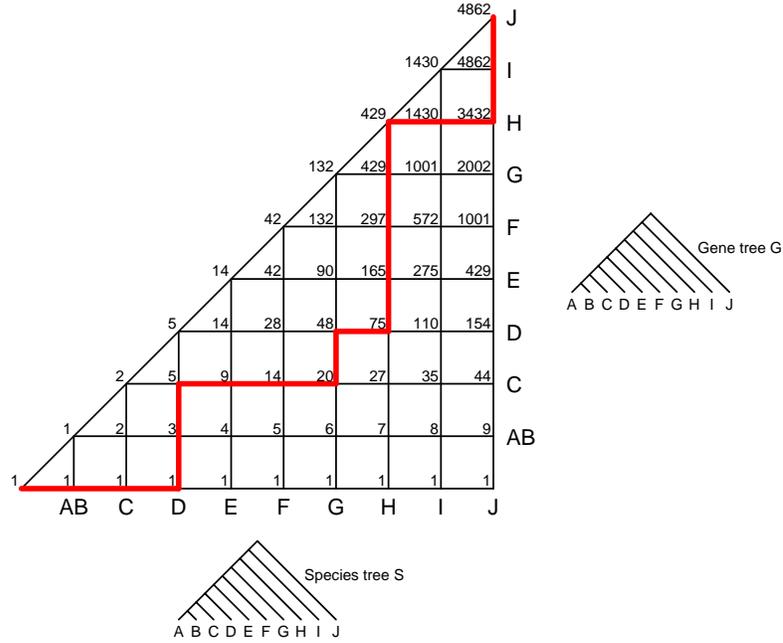}
\end{center}
\vskip -14.5cm
\caption{The correspondence between monotonic paths that do not cross above the $y=x$ diagonal of an $(n-1) \times (n-1)$ square lattice and coalescent histories for a matching caterpillar gene tree and species tree with $n=10$ leaves. The lower left corner represents the origin $(0,0)$. Monotonic paths from $(0,0)$ to $(i,j)$ represent the partial coalescent histories $A_{i,j}$ for $(G_j,S_i)$. Values $|A_{i,j}|$ are taken from eq.~\ref{eqRecursionCatalan}, using $(i,j)$ in place of $(n,k)$. Species tree internal edges are read from left to right: $AB$ labels the species tree internal edge from which $A$ and $B$ descend, and each successive label indicates the internal edge ancestral both to the leaf corresponding to the associated label and to the caterpillar subtree containing all prior labels. Gene tree internal nodes are read in the same manner from bottom to top. The monotonic path shown in red indicates the locations on the species tree of the gene tree coalescences of a specific coalescent history.}
\label{figLatticeMatching}
\end{figure}
%%%%%%%%%%%%%%%%%%%%%%%%%%%%%%%%%%%%%%%%%%%%%%%%%%%%%%%%%%%%%%%%%%
%%%%%%%%%%%%%%%%%%%%%%%%%%%%%%%%%%%%%%%%%%%%%%%%%%%%%%%%%%%%%%%%%%
%%%%%%%%%%%%%%%%%%%%%%%%%%%%%%%%%%%%%%%%%%%%%%%%%%%%%%%%%%%%%%%%%%

The construction takes advantage of the caterpillar shape of both gene tree and species tree. Because internal nodes of a caterpillar tree can be placed in order with each entry descended from the next until the root is reached, simply stating the next leaf label suffices to specify the leaves descended from the next internal node. Movement from left to right in Figure \ref{figLatticeMatching} indicates movement from the cherry of the species tree toward the root, and movement from bottom to top indicates coalescence in the gene tree.

%%%%%%%%%%%%%%%%%%%%%%%%%%%%%%%%%%%%%%%%%%%%%%%%%%%%%%%%%%%%%%%%%%%%%%%
\subsection{Non-matching gene trees and species trees}
\label{secNonmatching}

\noindent Our key insight is that a version of the construction of \cite{Degnan05} linking coalescent histories and monotonic paths applies even if the gene tree and species tree are non-matching, provided that both continue to be caterpillars. Coalescent histories for non-matching caterpillars can be associated with \emph{roadblocked} monotonic paths that do not cross above the $y=x$ diagonal of an $(n-1) \times (n-1)$ square lattice.

\begin{defi}
In a lattice, a \emph{roadblocked monotonic path} is a monotonic path that is not permitted to pass through certain specified lattice points. We term these lattice points \emph{roadblocks}.
\label{defiRoadblocked}
\end{defi}

Consider a caterpillar gene tree $G$ and a caterpillar species tree $S$, whose leaves are both bijectively associated with the same set of $n$ leaves, but that do not necessarily match. As in Section \ref{secMatching}, we associate points on the x-axis of an $(n-1) \times (n-1)$ lattice with species tree internal edges in $S$, and we associate points on the y-axis with gene tree internal nodes in $G$. We continue to label internal nodes of $G$ and $S$ in increasing order from 1 to $n-1$, from the cherry to the root, indexing the gene tree internal nodes by $j$ and the species tree internal nodes by $i$.

As is true in the matching case, for each $j$ from 1 to $n-1$, each coalescent history must have $h(j) \geqslant j$, as a gene tree internal node $j$ must map to a species tree internal edge ancestral to at least as many leaves as descend from node $j$ in $G$. Hence, each coalescent history for $(G,S)$ corresponds to a monotonic path that has $j \leqslant i$ and hence does not cross the $y=x$ diagonal of the lattice. However, an additional constraint is imposed by the fact that $G$ and $S$ do not necessarily match.

Given $G$ and $S$, let $\pi(G)$ denote the permutation of the gene tree leaf labels $\mathbf{g} = (g_1, g_2, \ldots, g_n)$ represented by the species tree leaf labels $\mathbf{s} = (s_1, s_2, \ldots, s_n)$. The action of $\pi$ sends the vector of leaf labels from one $n$-tuple to another, and we denote the index in $S$ of $g_k$, the $k$th label of $G$, by $\pi_k(G)$.

For the leaf labels $g_1, g_2, \ldots, g_n$ in $G$, let $f(g_k)$ denote the minimal internal edge of $S$ ancestral to leaf $s_{\pi_k(G)}$, the species tree leaf with label $g_k$. For a matching gene tree and species tree $(G,S)$, $\pi$ is the identity permutation so that $\pi_k(G)=k$; we then have $f(g_1)=f(g_2)=1$, and $f(g_k)=k-1$ for $3 \leqslant k \leqslant n$.

For general $(G,S)$ that do not necessarily match, by Definition \ref{defiCoalhist}, (i) if $k=1$ or $k=2$, then $f(g_k)=\max_{\ell \in \{1,2\}} \pi_\ell(G)-1$, and (ii) if $3 \leqslant k \leqslant n$, then $f(g_k)=\max_{\ell \in \{1,2,\ldots,k\}} \pi_\ell(G)-1$. This rule encodes the fact that a gene tree coalescence can occur only on a species tree edge ancestral to all species tree leaves labeled by the elements of the set of labels for leaves descended from the gene tree coalescence.

Consider the partial coalescent histories $A_{i,j}$ with $i \geqslant j$. As in Section \ref{secMatching}, for $j \geqslant 1$, $|A_{i,j}|=0$ for all $(i,j)$ with $i<j$. For each $j$ from 1 to $n-1$, the minimal internal edge of $S$ that is ancestral to all leaves labeled by labels of leaves of $G$ that descend from gene tree internal node $j$ is $f(g_{j+1})$. Therefore, for $j \geqslant 1$, we have $|A_{i,j}|=0$ for all $(i,j)$ with $i < f(g_{j+1})$. Note that these $(i,j)$ are the only roadblocks: for $j \geqslant 1$, $f(g_{j+1}) \geqslant j$, as $f(g_{j+1})$ is one less than the maximum of $j+1$ distinct elements of $\{1, 2, \ldots, n-1\}$, a quantity greater than or equal to $j$. For $j \geqslant 1$, because $|A_{i,j}|=0$ for all lattice points $(i,j)$ with $i < f(g_{j+1})$, all such points are roadblocks.

We also note that for $1 \leqslant j \leqslant j^\prime \leqslant n-1$, $f(g_{j^\prime+1}) \geqslant f(g_{j+1})$. The set of descendant leaves of internal node $j^\prime+1$ of $G$ contains as a subset the descendant leaves of internal node $j+1$ of $G$. Hence, the minimal internal edge of $S$ ancestral to all labels that label leaves descended from internal node $j^\prime+1$ of $G$ has an index at least as great as the corresponding internal edge of $S$ associated with internal node $j+1$ of $G$. Consequently, if $(i,j)$ is a roadblock, then because $i < f(g_{j+1})$ and $f(g_{j^\prime+1}) \geqslant f(g_{j+1})$ for $j^\prime \geqslant j$, we can conclude that $(i,j^\prime)$ is a roadblock for each $j^\prime$ with $j \leqslant j^\prime \leqslant i$.

As in Section \ref{secMatching}, each partial coalescent history in $A_{i,j}$ is formed in one of two ways. For $j\geqslant 1$, gene tree node $j$ is mapped either to species tree internal edge $i$, or to one of the edges $1, 2, \ldots, i-1$. The former case produces $|A_{i,j-1}|$ partial coalescent histories, and the latter produces $|A_{i-1,j}|$. Hence, the recursion $|A_{i,j}| = |A_{i,j-1}| + |A_{i-1,j}|$ is still satisfied. We still have the constraints $|A_{i,j}|=0$ for $j \geqslant 1$ and $i<j$, $|A_{i,0}|=1$ for $j=0$ and $0 \leqslant i \leqslant n-1$, and $|A_{i,j}|=0$ for all $(i,j)$ that do not satisfy $0 \leqslant i,j \leqslant n-1$. We also have the new constraint $|A_{i,j}|=0$ for all $(i,j)$ that satisfy $i < f(g_{j+1})$.

The set of roadblocks for $(G,S)$ is defined by $B_{G,S} = \{(i,j) \, | \, 1 \leqslant j \leqslant i \leqslant n-1 \text{ and } i < f(g_{j+1}) \}$. We have therefore demonstrated the following proposition.

\begin{prop}
Consider a caterpillar gene tree $G$ and a caterpillar species tree $S$, both bijectively associated with the same set of $n$ leaf labels, but that do not necessarily match. Then $(G,S)$ can be associated with a set of roadblocks $B_{G,S}$ such that the coalescent histories for $(G,S)$ bijectively correspond to roadblocked monotonic paths that do not cross the $y=x$ diagonal of an $(n-1) \times (n-1)$ lattice.
\label{propBijection}
\end{prop}

By definition of $B_{G,S}$, we immediately see that if $(i,j)$ is a roadblock for $1 \leqslant j \leqslant i \leqslant n-1$, then $(k,j)$ is a roadblock as well for each $k$ with $j \leqslant k \leqslant i$. We can also see that if $(i,j)$ is a roadblock for $1 \leqslant j \leqslant i \leqslant n-1$, then $(i, \ell)$ is a roadblock as well for each $\ell$ with $j \leqslant \ell \leqslant n-1$; this result follows from the fact that $f(g_{j^\prime + 1}) \geqslant f(g_{j+1})$ for $1 \leqslant j \leqslant j^{\prime} \leqslant n-1$. We have the following remark.
\begin{rem}
Consider a caterpillar gene tree $G$ and a caterpillar species tree $S$. The roadblock set $B_{G,S}$ consists of a set of points $(i,j)$ with $1 \leqslant j \leqslant i \leqslant n-1$ such that if $(i,j) \in B_{G,S}$, then (i) $(k,j) \in B_{G,S}$ for all $k$ with $j \leqslant k \leqslant i$, and (ii) $(j, \ell) \in B_{G,S}$ for all $\ell$ with $j \leqslant \ell \leqslant i$.
\label{remBGS}
\end{rem}

Figure \ref{figLatticeNonmatching} illustrates the correspondence between coalescent histories and roadblocked monotonic paths. In Figure \ref{figLatticeNonmatching}, we have $(f(g_1), f(g_2), f(g_3), f(g_4), f(g_5), f(g_6), f(g_7), f(g_8), f(g_9)) = (5,5,5,5,6,8,8,9,9)$. Because $f(g_{1+1})=5$, $(4,1)$ is a roadblock, as are
$(3,1)$, $(2,1)$, and $(1,1)$ for the same reason ($(i,j)$ is a roadblock if $j \leqslant i < f(g_{j+1})$). Because $f(g_{2+1})=5$, $(4,2)$ is also a roadblock, as are $(3,2)$ and $(2,2)$. We can also identify $(4,2)$, $(3,2)$, and $(2,2)$ by Remark \ref{remBGS} as roadblocks as a consequence of the fact that $(4,1)$, $(3,1)$, and $(2,1)$ are roadblocks. Continuing through all $(i,j)$, we identify 15 roadblocks in Figure \ref{figLatticeNonmatching}.

%%%%%%%%%%%%%%%%%%%%%%%%%%%%%%%%%%%%%%%%%%%%%%%%%%%%%%%%%%%%%%%%%%
%%%%%%%%%%%%%%%%%%%%%%%%%%% Figure 4 %%%%%%%%%%%%%%%%%%%%%%%%%%%%%
%%%%%%%%%%%%%%%%%%%%%%%%%%%%%%%%%%%%%%%%%%%%%%%%%%%%%%%%%%%%%%%%%%
\begin{figure}[tpb]
\vskip -1.4cm
\begin{center}
\includegraphics[width=7in]{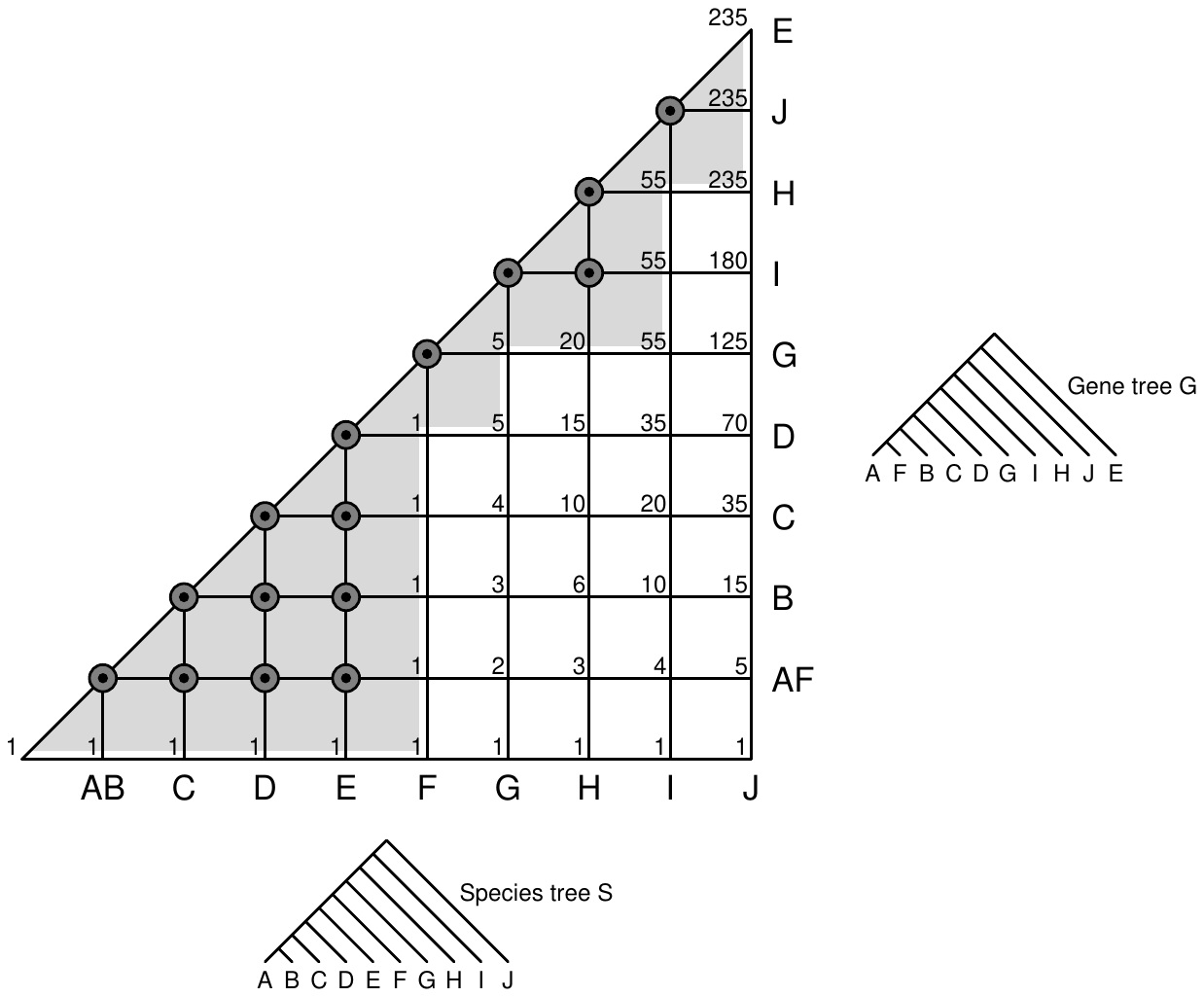}
\end{center}
\vskip -14.5cm
\caption{The correspondence between monotonic paths that do not cross above the $y=x$ diagonal of an $(n-1) \times (n-1)$ square lattice and coalescent histories for a non-matching caterpillar gene tree and species tree with $n=10$ leaves. Roadblocks are indicated by circles on lattice points; no roadblocked monotonic paths traverse the shaded regions. The lower left corner represents the origin $(0,0)$. Monotonic paths from $(0,0)$ to $(i,j)$ represent the partial coalescent histories $A_{i,j}$ for $(G_j,S_i)$.  Values $|A_{i,j}|$ are taken from eq.~\ref{eqRecursionCatalan}, using $(i,j)$ in place of $(n,k)$. Species tree internal edges are read from left to right: $AB$ labels the species tree internal edge from which $A$ and $B$ descend, and each successive label indicates the internal edge ancestral both to the leaf corresponding to the associated label and to the caterpillar subtree containing all prior labels. Gene tree internal nodes are read in the same manner from bottom to top.}
\label{figLatticeNonmatching}
\end{figure}
%%%%%%%%%%%%%%%%%%%%%%%%%%%%%%%%%%%%%%%%%%%%%%%%%%%%%%%%%%%%%%%%%%
%%%%%%%%%%%%%%%%%%%%%%%%%%%%%%%%%%%%%%%%%%%%%%%%%%%%%%%%%%%%%%%%%%
%%%%%%%%%%%%%%%%%%%%%%%%%%%%%%%%%%%%%%%%%%%%%%%%%%%%%%%%%%%%%%%%%%

From Proposition \ref{propBijection}, we immediately obtain that the number of coalescent histories for $(G,S)$ is given by the number of roadblocked monotonic paths that do not cross above the $y=x$ diagonal of an $(n-1) \times (n-1)$ lattice, where the roadblocks are those in the set $B_{G,S}$. We also obtain a simple proof of the following corollary, which appeared as Remark 15 of \cite{DegnanAndRhodes15}.

\begin{coro} Consider a caterpillar gene tree topology $G$ and a caterpillar species tree topology $S$. The number of coalescent histories for $(G,S)$ is strictly greater for $G=S$ than for each choice of $G \neq S$.
\label{coroStrictlyGreater}
\end{coro}
\begin{proof} By Proposition \ref{propBijection}, coalescent histories for $(G,S)$ correspond to roadblocked monotonic paths that do not cross the $y=x$ diagonal of an $(n-1) \times (n-1)$ lattice.

In the case that $G=S$, applying Lemma \ref{lemBijection}, the number of coalescent histories is the number of monotonic paths that do not cross the $y=x$ diagonal of the lattice.

Adding a roadblock to the lattice necessarily reduces the number of monotonic paths from $(0,0)$ to $(n-1,n-1)$, as each lattice point has at least one monotonic path that passes through it. Because the number of coalescent histories for $(G,S)$ is equal to the number of roadblocked monotonic paths on the lattice, it suffices to show that for $G \neq S$, at least one lattice point is a roadblock.

Because $G \neq S$, there exists some internal node $j$ of $G$ at least one of whose descendant leaves has a label not contained in the label set of the leaves descended from internal node $j$ of $S$.  This leaf has $j < f(g_{j+1})$. Hence, $(j,j)$ is a roadblock, and $(G,S)$ is associated with fewer monotonic paths than is $(S,S)$.
\end{proof}

%%%%%%%%%%%%%%%%%%%%%%%%%%%%%%%%%%%%%%%%%%%%%%%%%%%%%%%%%%%%%%%%%%%%%%%
\subsection{Roadblock sets}
\label{secRoadblockSets}

\noindent Given a caterpillar species tree $S$, Remark \ref{remBGS} suggests a characterization of the possible sets of roadblocks, considering all caterpillar gene trees $G$. Each roadblock set has the property that within a row, all points to the left of a roadblock and on or below the $y=x$ diagonal are also roadblocks. Within a column, all points above a roadblock and on or below the $y=x$ diagonal are roadblocks.

\begin{prop}
Consider a caterpillar species tree topology $S$ with $n$ leaves. For each caterpillar gene tree topology $G$ with $n$ leaves, denote its associated roadblock set by $B_{G,S}$. Considering all $n!/2$ possible caterpillar gene tree topologies, the distinct roadblock sets are bijectively associated with the $C_{n-1}$ monotonic paths on the $(n-1) \times (n-1)$ lattice that do not cross the $y=x$ diagonal.
\label{propRoadblocks}
\end{prop}

\begin{proof}
Consider a roadblock set $B_{G,S}$. For each $i$ from 1 to $n-2$, we identify the largest $j$ such that $(i,j)$ is not a roadblock. Call this value $j_i$. A unique monotonic path connects $(0,0), (1,j_1), (2,j_2), \ldots, (n-2,j_{n-2}), (n-1,n-1)$: by Remark \ref{remBGS}, for each $i$ and each $j > j_i$, $(i,j)$ is either a roadblock or it lies above the $y=x$ line. Hence, denoting $j_0=0$ and $j_{n-1}=n-1$, for each $i$ from 1 to $n-1$, a monotonic path from $(i-1,j_{i-1})$ to $(i,j_i)$ must proceed horizontally by length  1 and then vertically by length $j_i - j_{i-1}$.

To show that this construction is injective, note that distinct monotonic paths are associated with distinct roadblock sets: consider a point $(i,j_i)$ appearing in one monotonic path $P_1$ but not in another one, $P_2$. Because $j_i$ is the largest value of $j$ that is not a roadblock for path $P_1$, $(i,j_i)$ must be a roadblock for $P_2$.

For surjectivity, consider a monotonic path from $(0,0)$ to $(n-1,n-1)$ that does not cross the $y=x$ line. For each $(i,j)$ in the path, $1 \leqslant i \leqslant n-2$, we assign each point $(i, \ell)$ with $j < \ell \leqslant i$ to be a roadblock.
\end{proof}

%%%%%%%%%%%%%%%%%%%%%%%%%%%%%%%%%%%%%%%%%%%%%%%%%%%%%%%%%%%%%%%%%%
%%%%%%%%%%%%%%%%%%%%%%%%%%% Figure 5 %%%%%%%%%%%%%%%%%%%%%%%%%%%%%
%%%%%%%%%%%%%%%%%%%%%%%%%%%%%%%%%%%%%%%%%%%%%%%%%%%%%%%%%%%%%%%%%%
\begin{figure}[tpb]
\vskip -1.5cm
\begin{center}
\includegraphics[width=7in]{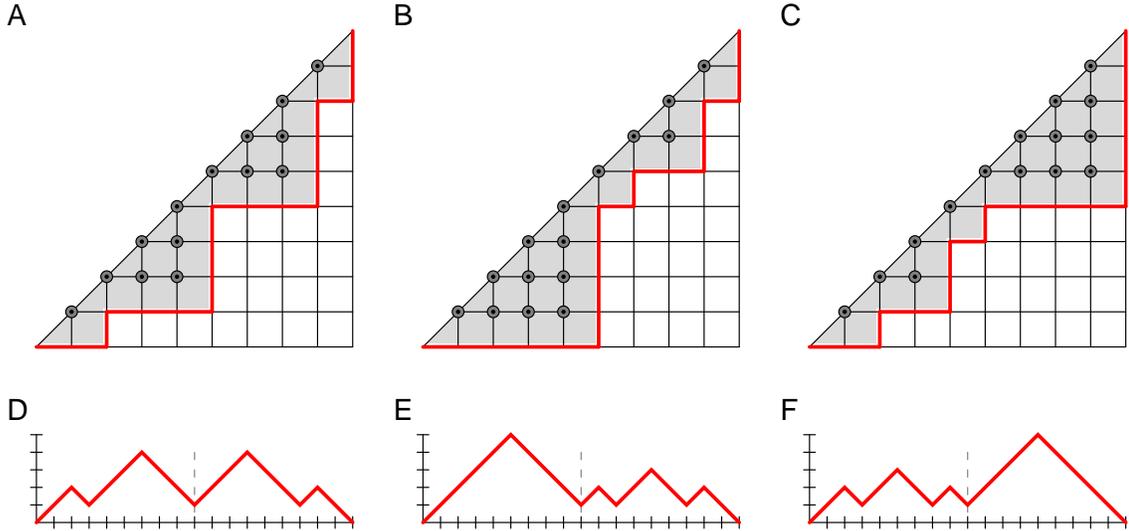}
\end{center}
\vskip -15.5cm
\caption{The correspondence between roadblock sets, monotonic paths that do not cross above the $y=x$ diagonal of an $(n-1) \times (n-1)$ square lattice, and Dyck paths of semi-length $n-1$. Given a roadblock set, the associated monotonic path is constructed by identifying for each $x$ coordinate from 0 to $n-1$ the lattice point of greatest $y$ coordinate, and then constructing the unique monotonic path through those points. Similarly, given a monotonic path, its roadblock set is obtained by placing roadblocks at each lattice point above and to the left of the path. (A) Roadblock set symmetric across the line $y=n-1-x$. (B) Roadblock set asymmetric across the line $y=n-1-x$. (C) Roadblock set asymmetric across the line $y=n-1-x$, obtained by reflecting the roadblocks in (B) over this line.  (D) Symmetric Dyck path associated with the roadblock set in (A). (E) Asymmetric Dyck path associated with the roadblock set in (B). (F) Asymmetric Dyck path associated with the roadblock set in (C), obtained by reversing the Dyck path in (E). The roadblock sets in (B) and (C) both generate 235 monotonic paths from $(0,0)$ to $(9,9)$.}
\label{figRoadblockSetSymmetric}
\end{figure}
%%%%%%%%%%%%%%%%%%%%%%%%%%%%%%%%%%%%%%%%%%%%%%%%%%%%%%%%%%%%%%%%%%
%%%%%%%%%%%%%%%%%%%%%%%%%%%%%%%%%%%%%%%%%%%%%%%%%%%%%%%%%%%%%%%%%%
%%%%%%%%%%%%%%%%%%%%%%%%%%%%%%%%%%%%%%%%%%%%%%%%%%%%%%%%%%%%%%%%%%

Figure \ref{figRoadblockSetSymmetric} provides an illustration of Proposition \ref{propRoadblocks}, showing how the monotonic path associated with a roadblock set is constructed and vice versa. The monotonic path associated with a roadblock set can be viewed as the monotonic path that comes as close as possible to the roadblocks. The roadblock set for a monotonic path is the set of points above and to the left of the path.

The number of distinct caterpillar trees is $n!/2$, whereas the number of distinct roadblock sets is the smaller $C_{n-1}$. For a given caterpillar species tree, we can place the $n!/2$ caterpillar gene trees into equivalence classes, where two gene trees are said to be \emph{history-equivalent} if and only if they are associated with the same roadblock set. Two history-equivalent caterpillar trees $G_1$ and $G_2$ have the same set of roadblocks and the same set of monotonic paths, and hence, the same set of coalescent histories, up to permutation of the leaf labels. These equivalence classes were termed \emph{history classes} by \cite{RosenbergAndTao08}, so that two caterpillars with the same roadblocks are in the same history class.

By Proposition \ref{propRoadblocks}, for a fixed species tree, the number of history classes considering all caterpillar trees is $C_{n-1}$; this result accords with the computation of 5 history classes for $n=4$ \citep[][Table V]{Rosenberg02} and 14 for $n=5$ \cite[][Table 3]{RosenbergAndTao08}. We have also seen in Corollary \ref{coroStrictlyGreater} that $C_{n-1}$ is the largest possible number of coalescent histories for a pair of caterpillar trees. We now ask how many of the values $1,2, \ldots, C_{n-1}$ can be the number of coalescent histories for some caterpillar gene tree and species tree.
The simplest upper bound on this quantity is $C_{n-1}$. To improve on this bound, it is convenient to use the bijection between monotonic paths that do not cross the $y=x$ diagonal of the $(n-1) \times (n-1)$ lattice and \emph{Dyck paths of semilength $n-1$}~\cite[][Corollary 6.3.2]{Stanley99}. Each monotonic path represents a series of steps by $(1,0)$ or $(0,1)$ from $(0,0)$ to $(n-1,n-1)$, with $x \geqslant y$ at each step. Each Dyck path represents a series of steps by $(1,1)$ or $(1,-1)$ from $(0,0)$ to $(n-1,0)$, with $y \geqslant 0$ at each step. The coalescent histories for $(G,S)$ can therefore be associated with Dyck paths, where each up-step represents addition of a species in the species tree and each down-step represents a gene tree coalescence.

A Dyck path of semi-length $n-1$ has $2n-2$ total up-steps and down-steps. The steps of Dyck paths can be written as a sequence, with $U$ denoting up-steps and $D$ denoting down-steps. A Dyck path can be \emph{reversed} in the following manner: we take the sequence of $U$ and $D$ steps in the path, reverse the order of steps, and exchange the positions of $U$ and $D$ steps. Thus, a path $UUUDUDDUDD$ becomes $UUDUUDUDDD$. Reversing a Dyck path corresponds to traversing the path in reverse order. A reversed Dyck path is itself a Dyck path; if the sequence of $U$ and $D$ steps in a Dyck path is reversed, then $y \leqslant 0$ at each step; exchanging the positions of the $U$ and $D$ steps reflects the path over the $y=0$ axis.

\begin{lemma}
Consider a caterpillar species tree topology $S$ with $n$ leaves. Consider gene tree topologies $G_1$ and $G_2$ such that $(i,j)$ is in the roadblock set $B_{G_1,S}$ if and only if $(n-1-j,n-1-i)$ is in the roadblock set $B_{G_2,S}$. Then $(G_1,S)$ and $(G_2,S)$ have the same number of coalescent histories.
\label{lemSymmetry}
\end{lemma}

\begin{proof}
We show that the coalescent histories for $(G_1,S)$ can be bijectively associated with the coalescent histories for $(G_2,S)$. Consider a coalescent history for $(G_1,S)$. Identify its associated monotonic path $M_1$ according to Proposition \ref{propBijection}, and identify the Dyck path $P_1$ associated with this monotonic path. Reverse $P_1$ to obtain $P_1^\prime$, and identify the monotonic path $M_1^\prime$ associated with $P_1^\prime$.

Because $M_1$ avoids each roadblock $(i,j)$ in $B_{G_1,S}$, after $i+j$ steps, $P_1$ cannot  have taken $i$ up-steps and $j$ down-steps. Because $P_1^\prime$ is the reverse of $P_1$, after $2n-2-i-j$ steps, $P_1^\prime$ cannot have taken $n-1-j$ up-steps and $n-1-i$ down-steps. The monotonic path $M_1^\prime$ therefore avoids the point $(n-1-j,n-1-i)$ for each roadblock $(i,j)$ in $B_{G_1,S}$. Hence, $M_1^\prime$ avoids each roadblock in $B_{G_2,S}$, and it therefore represents a coalescent history for $G_2$. Similarly, beginning from the coalescent history for $(G_2,S)$ associated with $M_1^\prime$, we find that $M_1$ represents a coalescent history for $B_{G_1,S}$.
\end{proof}

The lemma demonstrates that for two roadblock sets, if the roadblocks of one can be obtained by transforming each roadblock $(i,j)$ of one into a roadblock $(n-1-j,n-1-i)$ of the other, then the associated caterpillar gene trees have the same number of coalescent histories.

Consider a set of points $B$ on or below the $y=x$ diagonal of the first quadrant of the $(n-1) \times (n-1)$ lattice (and not on lines $y=0$ or $x=n-1$) with the property that if $(i,j) \in B$, then $(k,j) \in B$ for all $k$ with $j \leqslant k \leqslant i$ and $(i,\ell) \in B$ for all $\ell$ with $j \leqslant \ell \leqslant i$. By Proposition \ref{propRoadblocks}, given a caterpillar species tree topology, $B$ is the roadblock set for some caterpillar gene tree. We term such a set a \emph{caterpillar-friendly} roadblock set.

\begin{defi}
Consider a caterpillar-friendly roadblock set $B$ for the $(n-1) \times (n-1)$ lattice. We say that $B$ is \emph{symmetric} if for each $(i,j) \in B$, $(n-1-j,n-1-i)$ is also in $B$. Otherwise, $B$ is \emph{asymmetric}.
\label{defiSymmetric}
\end{defi}

\noindent In a symmetric caterpillar-friendly roadblock set, when the points in the roadblock set are reflected across the line $y=n-1-x$, the same roadblock set is obtained (Figure \ref{figRoadblockSetSymmetric}A). For an asymmetric caterpillar-friendly roadblock set, a different roadblock set is obtained by this reflection (Figure \ref{figRoadblockSetSymmetric}B and \ref{figRoadblockSetSymmetric}C).

For the $(n-1) \times (n-1)$ lattice, denote by $Q_{n-1}$ and $R_{n-1}$ the numbers of symmetric and asymmetric caterpillar-friendly roadblock sets, respectively. By Lemma \ref{lemSymmetry}, the asymmetric caterpillar-friendly roadblock sets can be partitioned into disjoint pairs such that the associated caterpillar gene trees for the two entries in a pair give rise to the same number of coalescent histories. Hence, considering all caterpillar gene trees and species trees, the number of distinct values possible for the number of coalescent histories is bounded above by $Q_{n-1}+R_{n-1}/2$, or because $Q_{n-1}+R_{n-1}=C_{n-1}$, by $(C_{n-1}+Q_{n-1})/2$.

We obtain $Q_{n-1}$ by counting all ways of placing roadblocks $(i,j)$ with $i+j \leqslant n-1$. By symmetry we then assign points $(n-1-j,n-1-i)$ to be roadblocks as well. Because of the bijection between roadblock sets and monotonic paths (Proposition \ref{propRoadblocks}), each set of roadblocks $(i,j)$ with $i+j \leqslant n-1$ is bijectively associated with a monotonic path from $(0,0)$ to a point $(i,n-1-i)$ for some $i$ with $0 \leqslant i \leqslant n-1$.

\begin{lemma}
The value of $Q_{n-1}$ is ${n-1 \choose \lfloor (n-1)/2 \rfloor}$.
\label{lemSymmetric}
\end{lemma}
\begin{proof}
Using eq.~\ref{eqCatalansTriangle}, the number of monotonic paths from $(0,0)$ to $(i,n-1-i)$ for some $i$ with $0 \leqslant i \leqslant n-1$ is obtained by the sum
\begin{equation*}
\sum_{j=0}^{\lfloor (n-1)/2 \rfloor} {n-1 \choose j} - \sum_{j=1}^{\lfloor (n-1)/2 \rfloor} {n-1 \choose j-1}.
\end{equation*}
The first sum gives $2^{n-2} + \frac{1}{2}{n-1 \choose (n-1)/2}$ for odd $n$, and $2^{n-2}$ for even $n$. The second sum gives $2^{n-2} - \frac{1}{2}{n-1 \choose (n-1)/2}$ for odd $n$, and $2^{n-2} - {n-1 \choose \lfloor (n-1)/2 \rfloor}$ for even $n$. Combining these cases, the result follows.
\end{proof}

This result appeared in Bonin {\it et al.}~(2003, Theorem 2.5) \nocite{BoninEtAl03} as the number of number of distinct first halves for Dyck paths, and in Deng {\it et al.}~(2015, Theorem 4.2) \nocite{DengEtAl15} as the number of Dyck paths invariant under reversal.

\begin{prop}
The size of the set of values that can equal the number of coalescent histories for at least one pair $(G,S)$ consisting of an $n$-leaf caterpillar gene tree $G$ and an $n$-leaf caterpillar species tree $S$ is bounded above by
$T_{n-1}=(C_{n-1}+Q_{n-1})/2$, or
$$\frac{1}{2} \bigg[ \frac{{2n-2 \choose n-1}}{n} + {n-1 \choose \lfloor (n-1)/2 \rfloor} \bigg] .$$
\label{propSymmetricallyDistinct}
\end{prop}

\noindent This quantity, which appeared in a bijectively related context in Bonin {\it et al.}~(2003, Theorem 4.2), gives the number of distinct Dyck paths up to reversal. Numerical values of the formulas in Lemma \ref{lemSymmetric} and Proposition \ref{propSymmetricallyDistinct} are shown in Table \ref{tableSequences}.

%%%%%%%%%%%%%%%%%%%%%%%%%%%%%%%%%%%%%%%%%%%%%%%%%%%%%%%%%%%%%%%%%
\begin{table}[tb]
\begin{center}
{\fontsize{8}{13}\selectfont \caption{The number of distinct values possible for the number of coalescent histories of a caterpillar gene tree and a caterpillar species tree.\label{tableSequences}}
\begin{tabular}{p{0.7in}p{0.6in}p{1in}p{1in}p{1.1in}p{1.1in}}
\hline\hline
Number of leaves $n$ &
Number of distinct roadblock sets &
Number of roadblock sets associated with symmetric Dyck paths &
Number of roadblock sets associated with asymmetric Dyck paths &
Upper bound on the number of distinct values for the number of coalescent histories &
Exact number of distinct values for the number of coalescent histories \\[0.5ex]
Notation &
$C_{n-1}$ &
$Q_{n-1}$ &
$P_{n-1}$ &
$T_{n-1}$ & \\[0.5ex]
Formula & $\frac{{2n-2 \choose n-1}}{n}$ & ${n-1 \choose \lfloor (n-1)/2 \rfloor}$
& $C_{n-1}-Q_{n-1}$ & $(C_{n-1}+Q_{n-1})/2$ & \\[0.5ex]
OEIS record & A000108 & A001405 & & A007123 & \\[0.5ex] \hline
2 & 1        & 1    & 0       & 1       & 1 \\
3 & 2        & 2    & 0       & 2       & 2 \\
4 & 5        & 3    & 2       & 4       & 4 \\
5 & 14       & 6    & 8       & 10      & 10 \\
6 & 42       & 10   & 32      & 26      & 21 \\
7 & 132      & 20   & 112     & 76      & 56 \\
8 & 429      & 35   & 394     & 232     & 154 \\
9 & 1430     & 70   & 1360    & 750     & 440 \\
10 & 4862    & 126  & 4736    & 2494    & 1373 \\
11 & 16796   & 252  & 16544   & 8524    & 4310 \\
12 & 58786   & 462  & 58324   & 29624   & 13925 \\
\hline
\hline
\end{tabular}}
\end{center}
\end{table}
%%%%%%%%%%%%%%%%%%%%%%%%%%%%%%%%%%%%%%%%%%%%%%%%%%%%%%%%%%%%%%%%%%

%%%%%%%%%%%%%%%%%%%%%%%%%%%%%%%%%%%%%%%%%%%%%%%%%%%%%%%%%%%%%%%%%%%%%%%
\section{Non-recursive enumeration of coalescent histories}
\label{secNonrecursive}

\noindent With the correspondence between coalescent histories for non-matching caterpillars and roadblocked monotonic paths established, we now turn to enumerating the coalescent histories of possibly non-matching caterpillar gene trees and species trees. We can do so recursively by enumerating roadblocked monotonic paths according to Proposition \ref{propBijection}; we can also obtain a non-recursive formula by applying eq.~\ref{eqRecursion}.

Without loss of generality, considering the two subtrees immediately descended from the root of a tree, we treat the left subtree as having a number of leaves greater than or equal to that of the right subtree. The right subtree of a caterpillar tree then has a single leaf, so that in eq.~\ref{eqRecursion}, the right subtree $G_R$ always has exactly one leaf in each successive step of the recursion. Hence, the term $B_{G_R, T(G_R,S),k+d(G_R,S)}$, follows the base case of the recursion and is equal to 1. Eq.~\ref{eqRecursion}, describing the number of coalescent histories for a caterpillar gene tree $G$ and a species tree $S$, then reduces to
\begin{equation}
\label{eqRecursionCaterpillar}
B_{G,S,m} = \sum_{k=1}^m B_{G_L, T(G_L,S),k+d(G_L,S)},
\end{equation}
with initial condition $B_{G,S,m}=1$ for all $m$ when $G$ has a single leaf.

If $S$ is also a caterpillar tree with $n$ leaves, then we can iterate the recursion $n-1$ times, at each step reducing the size of the left subtree $G_L$ by one, until $G_L$ has a single leaf, the base case applies, and the summand equals 1. Each iteration introduces a new summation, with its upper limit depending on the associated $d(G_L,S)$, the number of edges that separate the root of $T(G_L,S)$ from the root of $S$. Continuing to label internal nodes of $G$ from 1 to $n-1$ in increasing order from the cherry to the root, we associate internal node $j$ of $G$ with index $k_{n-j}$. Setting the integer parameter $m$ equal to 1, we have
\begin{equation}
\label{eqRecursionCaterpillar2}
B_{G,S,1} = \sum_{k_1=1}^{1} \sum_{k_2=1}^{k_1+c_1} \sum_{k_3=1}^{k_2+c_2} \dots \sum_{k_{n-1}=1}^{k_{n-2}+c_{n-2}} 1,
\end{equation}
where the constant $c_j$ represents the number of additional edges of $S$ that are possible locations for gene tree coalescence $j$ but that are not possible for gene tree coalescence $j+1$.

For $1 \leqslant j \leqslant n-1$, consider gene tree internal node $j$. Let $L_j$ be the set of labels for all $j+1$ leaves descended from $j$. Following the definitions in eq.~\ref{eqRecursion}, let $T_j(G,S)$ denote the smallest subtree of $S$ that has the property that each label in $L_j$ labels one of its leaves, and let $d_j$ denote the number of edges separating the root of $T_j(G,S)$ from the root of $S$. Then $d_j+1$ gives the number of edges of $S$ on which gene tree coalescence $j$ can occur (the +1 represents the root edge of $S$). The quantity $u_j=n-1-j-d_j$, equal to the number of edges of $S$ ancestral to at least $j+1$ leaves (or $n-j$) but on which gene tree coalescence $j$ \emph{cannot} occur, represents the number of roadblocks $(i,j)$ with fixed $j$ and $i \geqslant j$.

For $j=1, 2, \ldots, n-2$, the desired quantity $c_j$, the number of additional edges of $S$ available for coalescence $j$ but not for coalescence $j+1$, equals $c_j=d_j-d_{j+1}$. We have therefore shown the following proposition.

\begin{prop}
Consider a caterpillar gene tree $G$ and a caterpillar species tree $S$, both bijectively associated with the same set of $n$ leaf labels, but that do not necessarily match. The number of coalescent histories for $(G,S)$ is obtained by eq.~\ref{eqRecursionCaterpillar2}, where the vector $(c_1,c_2,\ldots,c_{n-2})$ is obtained as a function $\mathbf{c}(G,S)$ that depends only on the topologies of $G$ and $S$.
\label{propNonrecursive}
\end{prop}

Note that if $G$ and $S$ match, then for each $j$ from 1 to $n-1$, $G_j=T_j(G,S)$, and hence $d_j=n-1-j$, $u_j=0$, and no roadblocks occur. We have $c_j=1$ for each $j$ from 1 to $n-2$, and eq.~\ref{eqRecursionCaterpillar2} becomes
$$\sum_{k_1=1}^{1} \sum_{k_2=1}^{k_1+1} \sum_{k_3=1}^{k_2+1} \dots \sum_{k_{n-1}=1}^{k_{n-2}+1} 1,$$
equal to the Catalan number $C_{n-1}$ \citep[][Theorem 3.4]{Rosenberg07:jcb}.

We take as an example the gene tree and species tree in Figure \ref{figLatticeNonmatching}. We report the values of the $u_j$, $d_j$ and $c_j$ in Table \ref{tableQuantities}. The number of coalescent histories is
$$\sum_{k_1=1}^{1} \sum_{k_2=1}^{k_1+0} \sum_{k_3=1}^{k_2+1} \sum_{k_4=1}^{k_3+0}
\sum_{k_5=1}^{k_4+2} \sum_{k_6=1}^{k_5+1} \sum_{k_7=1}^{k_6+0} \sum_{k_8=1}^{k_7+0} \sum_{k_9=1}^{k_8+0} 1 = 235.$$
We can also obtain this result by recursive summation of roadblocked monotonic paths (Figure \ref{figLatticeNonmatching}).

%%%%%%%%%%%%%%%%%%%%%%%%%%%%%%%%%%%%%%%%%%%%%%%%%%%%%%%%%%%%%%%%%
\begin{table}[tb]
\begin{center}
{\fontsize{8}{13}\selectfont \caption{Quantities associated with the enumeration of coalescent histories for a caterpillar gene tree $(((((((((A,F),B),C),D),G),I),H),J),E)$ and species tree $(((((((((A,B),C),D),E),F),G),H),I),J)$.\label{tableQuantities}}
\begin{tabular}{cccccccccc}
\hline\hline
Internal node index in gene tree $G$ ($j$)    & 9 & 8 & 7 & 6 & 5 & 4 & 3 & 2 & 1 \\
\hline
Summation index ($n-j$)      & 1 & 2 & 3 & 4 & 5 & 6 & 7 & 8 & 9 \\
Number of roadblocks ($u_j$) & 0 & 1 & 1 & 2 & 1 & 1 & 2 & 3 & 4 \\
Distance between root of $T_j(G,S)$ and root of $S$ ($d_j$)  & 0 & 0 & 1 & 1 & 3 & 4 & 4 & 4 & 4 \\
Nodes possible for coalescence $j$ but not $j+1$ ($c_j=d_j-d_{j+1}$) & NA & 0 & 1 & 0 & 2 & 1 & 0 & 0 & 0 \\
Summation term & $\sum\limits_{k_1=1}^{1}$ & $\sum\limits_{k_2=1}^{k_1+0}$ & $\sum\limits_{k_3=1}^{k_2+1}$ & $\sum\limits_{k_4=1}^{k_3+0}$ & $\sum\limits_{k_5=1}^{k_4+2}$ & $\sum\limits_{k_6=1}^{k_5+1}$ &
$\sum\limits_{k_7=1}^{k_6+0}$ & $\sum\limits_{k_8=1}^{k_7+0}$ & $\sum\limits_{k_9=1}^{k_8+0}$ \\
\hline
\hline
\end{tabular}}
\end{center}
\end{table}
%%%%%%%%%%%%%%%%%%%%%%%%%%%%%%%%%%%%%%%%%%%%%%%%%%%%%%%%%%%%%%%%%%

By exhaustive use of Proposition \ref{propNonrecursive}, we have evaluated all possible values of the number of coalescent histories for the $n!/2$ caterpillar gene tree topologies with $n$ leaves. This exhaustive computation applies eq.~\ref{eqRecursionCaterpillar2} with all possible vectors $(c_1,c_2,\ldots,c_{n-2})$ that correspond to gene trees---in other words, the $C_{n-1}$ vectors with $0 \leqslant \sum_{k=1}^j c_j \leqslant j$ for each $j$ from 1 to $n-2$~\cite[][replacing $a_i$ in item 81 with $1-c_i$]{Stanley15}.

The upper bound from Proposition \ref{propSymmetricallyDistinct} on the number of distinct values for the number of coalescent histories is relatively tight for small $n$, but already is more than double the exact computation for $n=12$ (Table \ref{tableSequences}). The smallest case in which the number of distinct values (21) differs from the upper bound (26) occurs with $n=6$ leaves, in which 1, 2, 3, 4, 5, 6, 7, 9, 10, 12, 13, 14, 16, 17, 19, 22, 23, 26, 28, 32, and 42 are achievable values for the number of coalescent histories. Values 5, 9, 10, 14, and 19 are each achieved with two distinct set of roadblocks that are not equivalent when reversing their associated Dyck paths.

Because $Q_{n-1} \ll C_{n-1}$, the upper bound for the number of distinct values for the number of coalescent histories of a caterpillar pair is asymptotically equivalent to $C_{n-1}/2$, half the maximum number of coalescent histories for caterpillars. Thus, although the number of caterpillars $n!/2$ grows much faster than the maximal number of coalescent histories $C_{n-1}$, asymptotically only at most half the values in the range of possible values for the number of coalescent histories are achieved by actual caterpillar gene trees.

%%%%%%%%%%%%%%%%%%%%%%%%%%%%%%%%%%%%%%%%%%%%%%%%%%%%%%%%%%%%%%%%%%
\section{Special families of caterpillar gene trees and species trees}
\label{secSpecial}

\noindent From Propositions \ref{propBijection} and \ref{propNonrecursive}, we can obtain a variety of corollaries that describe the number of coalescent histories for special pairs of non-matching caterpillar trees. For certain classes of pairs, the number of coalescent histories can be obtained in closed form.

%%%%%%%%%%%%%%%%%%%%%%%%%%%%%%%%%%%%%%%%%%%%%%%%%%%%%%%%%%%%%%%%%%%%%%%
\subsection{Nearest-neighbor-interchange}
\label{secNNI}

\noindent For a fixed caterpillar species tree $S$, we first consider caterpillar gene trees $G$ that differ from $S$ by a single nearest-neighbor-interchange move (NNI). We have the following result.
\begin{prop}
Consider a caterpillar species tree topology $S$ with $n$ leaves and a caterpillar gene tree topology $G$ that differs from $S$ by an NNI move. Then
(i) the roadblock set $B_{G,S}$ consists of a single point $(i,i)$ on the diagonal of the square lattice, for some $i$ with $1 \leqslant i \leqslant n-2$.
(ii) The number of coalescent histories for $(G,S)$ is $C_{n-1}-C_{i} C_{n-1-i}$.
\label{propNNI}
\end{prop}
\begin{proof}
We use the bijection between coalescent histories and roadblocked monotonic paths (Proposition \ref{propBijection}). We label leaves on the trees from $1$ to $n$ as in Section \ref{secNonmatching}, using the permutation $\pi$ to map the leaves of $G$ to the leaves of $S$. By Definition \ref{defiNNI}, an NNI move exchanges a single pair of leaves labeled $k$ and $k+1$ in $G$ for some $k \in \{2,3,\ldots,n-1\}$, or it exchanges leaves 1 and 3. Let $k_s$ be the smaller of the two labels for the leaves participating in the NNI move, and let $k_\ell$ be the larger of the two labels. We then have $\pi_{k_s}(G)=k_s+1$ and $\pi_{k_s+1}(G)=k_s$ if $k_s \in \{2,3,\ldots,n-1\}$, and $\pi_1(G)=3$ and $\pi_3(G)=1$ if $k_s=1$.

(i) Following Section \ref{secNonmatching}, for $(G,S)$ differing by one NNI move, the minimal internal edge of $S$ ancestral to leaf $k_s$ of $G$ is $f(g_{k_s})=k_s$ if $2 \leqslant k_s \leqslant n-1$, and $f(g_1)=2$ if $k_s=1$. The roadblocks $(i,j)$ in the square lattice are those points that satisfy $i < f(g_{j+1})$. By construction, $(k_s-1,k_s-1)$ is the only roadblock if $2 \leqslant k_s \leqslant n-1$, and $(1,1)$ is the only roadblock if $k_s=1$.

(ii) The number of coalescent histories for $(G,S)$ is the number of coalescent histories for the case of no roadblocks, or $C_{n-1}$ (Lemma \ref{lemBijection}), minus the number of monotonic paths from $(0,0)$ to $(n-1,n-1)$ that do not cross the diagonal and that pass through the roadblock. For a roadblock at $(i,i)$, this latter quantity is $C_i C_{n-1-i}$, multiplying the number of monotonic paths $C_i$ from $(0,0)$ to $(i,i)$ that do not cross the diagonal by the number of monotonic paths from $(i,i)$ to $(n-1,n-1)$ that do not cross the diagonal.
\end{proof}

Figure \ref{figNNI}A illustrates the result of Proposition \ref{propNNI} with a pair $(G,S)$ that differ by a single NNI move. The number of coalescent histories for the example is 4274, as obtained by Proposition \ref{propBijection}. Using Proposition \ref{propNNI}, we see that  $4274=C_9-C_4 C_5 = 4862 - 14 \times 42$.

Extending Proposition \ref{propNNI}, we can establish a formula for the number of coalescent histories when the permutation $\pi$ performs $k$ NNI moves, $1 \leqslant k \leqslant \lfloor \frac{n-1}{2} \rfloor$, in such a way that $\pi$ consists of disjoint cycles of length 2. The roadblock set for $(G,S)$ with such a permutation contains $k$ points on the diagonal of the $(n-1) \times (n-1)$ lattice. For this roadblock set, we count the monotonic paths that do not cross the diagonal by use of the inclusion--exclusion principle.

\begin{prop}
Consider a caterpillar species tree topology $S$ with $n$ leaves and a caterpillar gene tree topology $G$ that differs from $S$ by $k$ disjoint NNI moves. Then
(i) the roadblock set $B_{G,S}$ consists of $k$ distinct points $(i_j,i_j)$ on the diagonal of the square lattice, $j \in \{1,2,\ldots,k\}$, with $0 < i_1 < \ldots i_k < n-1$.
(ii) The number of coalescent histories for $(G,S)$ can be written
\begin{equation}
\label{eqInclusionExclusion}
C_{n-1}
+ \sum_{\ell=1}^{k} (-1)^{\ell}
\left(\sum_{ (j_1, j_2, \ldots, j_\ell) \in \{1,2,\ldots,k \}^\ell  \atop 0< i_{j_{1}}<...<i_{j_{\ell}} < n-1}
C_{i_{j_1}} C_{i_{j_2}-i_{j_1}} \cdots C_{i_{j_\ell}-i_{j_\ell -1 }} C_{n - 1 - i_{j_\ell}}\right).
\end{equation}
\label{propInclusion}
\end{prop}
\begin{proof}
As in Proposition \ref{propNNI}, we rephrase a problem of enumerating coalescent histories in the language of roadblocked monotonic paths.

(i) Because the $k$ NNI moves are disjoint, we can apply Proposition \ref{propNNI}i sequentially $k$ times, once for each NNI move. Each of the $k$ NNI moves is associated with a roadblock on the diagonal of the square lattice, the location of which is determined by the identity of the pair of leaves exchanged: if $k_s$ is the smaller of the two labels for leaves of $G$ participating in the move, then the roadblock location is $(k_s-1,k_s-1)$ for $2 \leqslant k_s \leqslant n-1$ and $(1,1)$ for $k_s=1$. Label the roadblocks $(i_1,i_1), (i_2,i_2), \ldots, (i_k,i_k)$, with $0 < i_1 < i_2 < \ldots < i_k < n-1$.

(ii) The total number of monotonic paths on the $(n-1) \times (n-1)$ square lattice is $C_{n-1}$ (Section \ref{secCatalan}). To obtain the desired quantity, we must subtract from $C_{n-1}$ the number of paths that pass through at least one of the $k$ roadblocks. By the inclusion--exclusion principle, this quantity can be written as a sum over nonempty subsets $\beta$ of the $k$ roadblocks of the number of monotonic paths that pass through all roadblocks in $\beta$. In particular, for each $j$ from 1 to $k$, denoting by $\beta_j$ the set of monotonic paths that pass through roadblock $(i_j,i_j)$, the number of monotonic paths that pass through at least one roadblock is
\begin{equation*}
\sum_{\ell=1}^k (-1)^{\ell+1} \left( \sum_{ (j_1, j_2, \ldots, j_\ell) \in \{1,2,\ldots,k \}^\ell  \atop 0< i_{j_{1}}<...<i_{j_{\ell}} < n-1} |\beta_{j_1} \cap \cdots \cap \beta_{j_\ell}| \right).
\end{equation*}
The cardinality of $|\beta_{j_1} \cap \cdots \cap \beta_{j_\ell}|$, representing the number of monotonic paths that pass through $(0,0)$, $(i_{j_1},i_{j_1})$, $\ldots$, $(i_{j_\ell},i_{j_\ell})$, and $(n-1,n-1)$, is a product of Catalan numbers, one for each pair of consecutive points that must be traversed:
$$C_{i_{j_1}} C_{i_{j_2}-i_{j_1}} \cdots C_{i_{j_\ell}-i_{j_\ell -1 }} C_{n - 1 - i_{j_\ell}}.$$ The result then follows.
\end{proof}

Figure \ref{figNNI}B illustrates Proposition \ref{propInclusion} for a case with three disjoint NNI moves. In this example, roadblocks in a 10-leaf tree appear at $(1,1)$, $(4,4)$, and $(9,9)$. By Proposition \ref{propInclusion}, the number of coalescent histories is
$C_9 - (C_1 C_8 + C_4 C_5 + C_8 C_1) + (C_1 C_3 C_5 + C_1 C_7 C_1 + C_4 C_4 C_1) - (C_1 C_3 C_4 C_1) = 2179$.

Proposition \ref{propNNI} additionally has the consequence that for a fixed caterpillar species tree $S$, the largest number of coalescent histories seen for a non-matching caterpillar gene tree $G$ occurs for a gene tree that differs from $S$ by a single NNI move.

\begin{coro}
\label{coroLargest}
Consider a caterpillar species tree topology $S$ with $n$ leaves. Considering all possible caterpillar gene tree topologies $G \neq S$,
(i) the largest number of coalescent histories for $(G,S)$ is obtained when the roadblock set $B_{G,S}$ consists of the single point $(\frac{n}{2}-1, \frac{n}{2}-1)$ or $(\frac{n}{2}, \frac{n}{2})$ for even $n$, or $(\frac{n-1}{2}, \frac{n-1}{2})$ for odd $n$. (ii) It equals
$$C_{n-1}-C_{\lfloor\frac{n-1}{2}\rfloor}C_{\lceil\frac{n-1}{2}\rceil}.$$
\end{coro}
\begin{proof}
First, note that for each pair $(G,S)$ whose roadblock set $B_{G,S}$ has more than one point, we can identify a pair $(G^\prime,S)$ whose roadblock set consists of a single point in $B_{G,S}$ and that hence has at least as many coalescent histories as $(G,S)$. By Remark \ref{remBGS}, the roadblock in a roadblock set consisting of a single point must be located on the diagonal.

Applying Proposition \ref{propNNI}ii, for the caterpillar gene tree topology $G$ that maximizes the number of coalescent histories with fixed $S$, that number of coalescent histories must equal $C_{n-1}-C_{i}C_{n-1-i}$ for some $i$ with $1 \leqslant i \leqslant n-2$. By Corollary 3.11 of \cite{Rosenberg07:jcb}, for fixed $n$, this quantity is maximized when $\{i, n-1-i\} = \{ \lfloor \frac{n-1}{2} \rfloor, \lceil \frac{n-1}{2} \rceil \}$.

From the proof of Proposition \ref{propNNI}, for fixed $S$, this maximum is associated with the gene tree topologies $G$ that differ from $S$ in that the leaves abutting the middle coalescence in the path from cherry to root of $G$ are transposed; the case of $n$ odd has one such coalescence and the case of $n$ even has two.
\end{proof}

%%%%%%%%%%%%%%%%%%%%%%%%%%%%%%%%%%%%%%%%%%%%%%%%%%%%%%%%%%%%%%%%%%
%%%%%%%%%%%%%%%%%%%%%%%%%%% Figure 6 %%%%%%%%%%%%%%%%%%%%%%%%%%%%%
%%%%%%%%%%%%%%%%%%%%%%%%%%%%%%%%%%%%%%%%%%%%%%%%%%%%%%%%%%%%%%%%%%
\begin{figure}[tpb]
\vskip -1.4cm
\begin{center}
\includegraphics[width=7.5in]{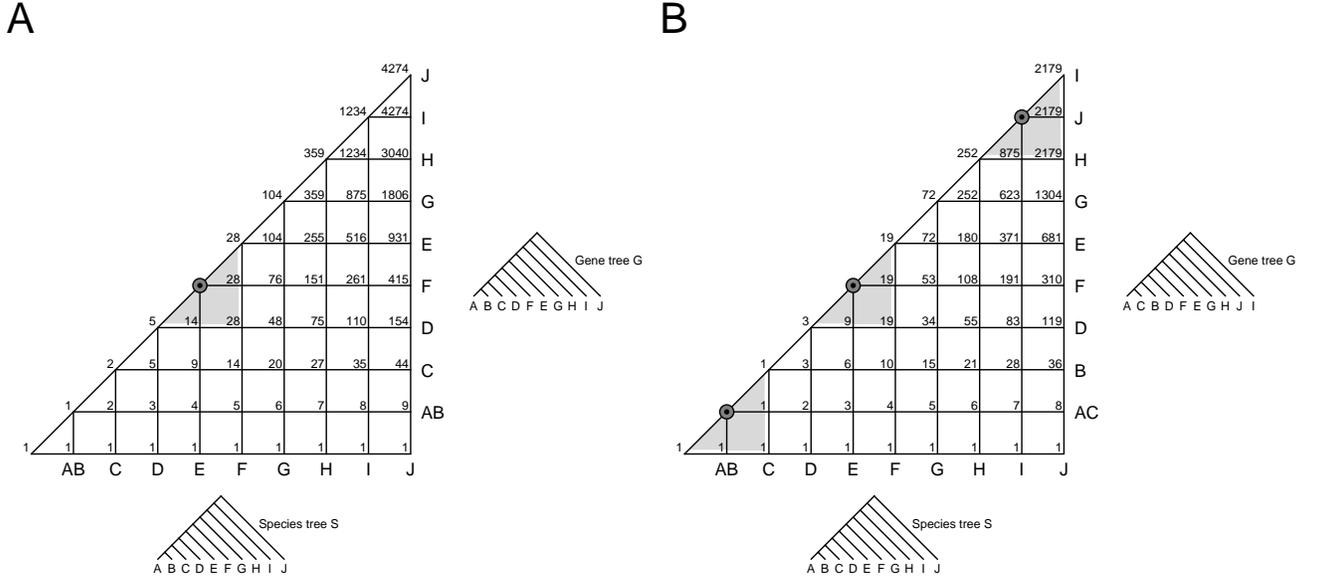}
\end{center}
\vskip -16cm
\caption{The number of coalescent histories for caterpillar gene tree topologies $G$ and species tree topologies $S$ that differ by nearest-neighbor-interchange (NNI) moves. (A) $G$ and $S$ differ by a single NNI move. (B) $G$ and $S$ differ by multiple disjoint NNI moves.}
\label{figNNI}
\end{figure}
%%%%%%%%%%%%%%%%%%%%%%%%%%%%%%%%%%%%%%%%%%%%%%%%%%%%%%%%%%%%%%%%%%
%%%%%%%%%%%%%%%%%%%%%%%%%%%%%%%%%%%%%%%%%%%%%%%%%%%%%%%%%%%%%%%%%%
%%%%%%%%%%%%%%%%%%%%%%%%%%%%%%%%%%%%%%%%%%%%%%%%%%%%%%%%%%%%%%%%%%

Corollary \ref{coroLargest} extends Corollary \ref{coroStrictlyGreater} by giving the exact number of coalescent histories for the pair $(G,S)$ that has the largest number of coalescent histories among non-matching caterpillars. For odd $n$, the number of coalescent histories in Corollary \ref{coroLargest} is $C_{n-1}-C_{(n-1)/2}^{2}$. For example, for $S=((((A,B),C),D),E)$ and $G=((((A,B),D),C),E)$, the number of coalescent histories is $C_4-C_2^2 = 14 - 2^2 = 10$ coalescent histories. For even $n$, the number of coalescent histories in Corollary \ref{coroLargest} is $C_{n-1}-C_{n/2-1} C_{n/2}$. For example, for $S=(((((A,B),C),D),E),F)$, both $G=(((((A,B),D),C),E),F)$ and $G=(((((A,B),C),E),D),F)$ have $C_5-C_2 C_3 = 42 - 2 \times 5 = 32$ coalescent histories. Figure \ref{figNNI}A gives the largest number of coalescent histories among nonmatching caterpillars with $n=10$ leaves.

We can quickly observe that the largest number of coalescent histories among  discordant caterpillar topologies grows at the same rate as the number of coalescent histories for matching caterpillars.

\begin{coro} Considering all non-matching caterpillar pairs $(G,S)$ with $n$ leaves, the largest number of coalescent histories for $(G,S)$ is asymptotic to $C_{n-1}$.
\label{eqCoroAsymptotic}
\end{coro}

\begin{proof}
Using Stirling's approximation, $n!\sim \sqrt{2\pi n}\left(\frac{n}{e}\right)^{n}$, we can verify that $C_{n} \sim 4^n / (n^{3/2} \sqrt{\pi} )$.

If $n$ is odd, then the largest number of coalescent histories for a non-matching pair satisfies
\begin{equation*}
C_{n-1}-C_{\frac{n-1}{2}}^2 \sim \frac{4^{n-1}}{(n-1)^{\frac{3}{2}} \sqrt{\pi}} - \frac{4^{n-1} 2^3}{(n-1)^3 \pi}.
\end{equation*}
If $n$ is even, then
\begin{equation*}
C_{n-1}-C_{\frac{n}{2}-1} C_{\frac{n}{2}} \sim \frac{4^{n-1}}{(n-1)^{\frac{3}{2}} \sqrt{\pi}} -  \frac{4^{n-1} 2^3}{(n^2-2n)^{\frac{3}{2}} \pi}.
\end{equation*}

In both cases, the leading term dominates, and $C_{n-1} - C_{\lfloor (n-1)/2 \rfloor} C_{\lceil (n-1)/2 \rceil} \sim C_{n-1}$.
\end{proof}

%%%%%%%%%%%%%%%%%%%%%%%%%%%%%%%%%%%%%%%%%%%%%%%%%%%%%%%%%%%%%%%%%%%%%%%
\subsection{Reverse incrementation of the leaf labels}
\label{secReverseCyclic}

\noindent Next, for a fixed caterpillar species tree $S$, we consider gene trees $G$ that differ from $S$ by incrementation.

Consider a caterpillar species tree $S$ with $n$ leaves and a caterpillar gene tree topology that differs from $S$ by an incrementation. By definition, the leaves of some component $G^\prime$ of $G$ and $S^\prime$ of $S$ differ by cyclic permutation. Recall that an incrementation with two labels is a NNI move.

\begin{prop} If $G$ is obtained by a reverse incrementation of $S$, then the roadblock set $B_{G,S}$ consists of a set of consecutive points on the diagonal of the square lattice.
\label{propReverse}
\end{prop}

\begin{proof}
Consider labels $k_s, k_\ell \in \{1,2,\ldots,n\}$, with $k_s < k_\ell$ and $k_\ell \neq 2$. By definition of reverse incrementation, for some component of $G$ with leaves sequentially labeled $k_s, k_s+1, \ldots, k_\ell$ from the cherry toward the root, associated leaves of $S$ are labeled $\pi_{k_s}(G)=k_s+1$, $\pi_{k_s+1}(G)=k_s+2, \ldots, \pi_{k_\ell-1}(G)=k_\ell, \pi_{k_\ell}(G)=k_s$.

As in the proof of Proposition \ref{propNNI}, we compute the minimal internal edge of $S$ ancestral to each leaf $g_k$ of $G$, $k \in \{k_s, k_s+1, \ldots, k_\ell \}$. We obtain $f(g_k)=k$ if $2 \leqslant k \leqslant n$ and $f(g_1)=2$ if $k=1$.

The roadblocks are the points $(i,j)$ satisfying $i < f(g_{j+1})$. We therefore find that the roadblocks are precisely those points $(k_s-1,k_s-1), \ldots, (k_\ell-2,k_\ell-2)$ if $k_s > 1$ and $(1,1), \ldots, (k_\ell-2,k_\ell-2)$ if $k_s=1$.
\end{proof}

Because all roadblocks lie on the diagonal under reverse incrementation, eq.~\ref{eqInclusionExclusion} can be applied to count coalescent histories. In the application of eq.~\ref{eqInclusionExclusion}, the distinct points on the diagonal through which monotonic paths cannot pass are $(k_s-1,k_s-1), \ldots, (k_\ell-2,k_\ell-2)$ if $k_s > 1$ and $(1,1), \ldots, (k_\ell-2,k_\ell-2)$ if $k_s=1$.

For example, in Figure \ref{figReverse}A, the reverse incrementation of leaf labels $C$, $D$, and $E$ has $k_s=3$ and $k_\ell=5$, so that the roadblocks lie at $(2,2)$ and $(3,3)$. The number of coalescent histories is obtained by eq.~\ref{eqInclusionExclusion} as $C_9 - (C_2 C_7 + C_3 C_6) + C_2 C_1 C_6 = 3608$.

If the reverse incrementation permutes all the labels, then all points $(1,1), \ldots (n-1,n-1)$ are roadblocks, and the number of coalescent histories is the number of monotonic paths not crossing the diagonal that lies one unit below the $y=x$ line (Figure \ref{figReverse}B). As the number of monotonic paths that do not pass above a diagonal of a square lattice, this computation gives $C_{n-2}$ coalescent histories. At the same time, the inclusion--exclusion computation of eq.~\ref{eqInclusionExclusion} produces a sum that traverses all subsets of the points $(1,1), \ldots, (n-1,n-1)$.

Thus, by use of eq.~\ref{eqInclusionExclusion}, this construction gives a combinatorial proof of a Catalan number identity.
\begin{coro} The Catalan number $C_{n-2}$ can be written as an alternating sum of products of Catalan numbers, where the sum proceeds over all compositions of $n-1$:
$$C_{n-2} = \sum_{k=1}^{n-1} (-1)^{k+1} \sum_{(v_1, \ldots, v_k) \in \{ \mathbf{v} : \sum_{i=1}^k v_i = n-1 \} } \prod_{i=1}^k C_{v_k}$$
\label{coroCompositions}
\end{coro}
\noindent This identity can be seen as counting Dyck paths of semi-length $n-1$ with no internal returns to the origin in two ways. $C_{n-2}$ gives the number of Dyck paths of semi-length $n-2$, as a Dyck path of length $n-1$ with no internal returns begins with an up-step that is followed by a Dyck path of semi-length $n-2$ and then a a down-step. The right-hand side instead uses the inclusion--exclusion principle to perform the computation by excluding Dyck paths of semi-length $n-1$ that have at least one return to the origin.

Interestingly, a reverse cycle that permutes all labels, even if it is not an incrementation, gives a Catalan number of coalescent histories, as it generates a roadblock set that consists of one or more diagonal lines. For example, with $S=(((((((((A,B),C),D),E),F),G),H),I),J)$, the reverse incrementation $G=(((((((((B,C),D),E),F),G),H),I),J),A)$
gives $C_8=1430$ coalescent histories (Figure \ref{figReverse}B), the reverse cycle $G=(((((((((C,D),E),F),G),H),I),J),A),B)$ gives $C_7=429$ coalescent histories, the reverse cycle $G=(((((((((D,E),F),G),H),I),J),A),B),C)$ gives $C_6=132$ coalescent histories, and so on.

We note also that eq.~\ref{eqInclusionExclusion} continues to apply if $S$ differs from $G$ by multiple disjoint reverse incrementations, as in Figure \ref{figReverse}C, which adds a two-leaf incrementation---an NNI move---to Figure \ref{figReverse}A. In this case, the number of coalescent histories is $C_9 - (C_2 C_7 + C_3 C_6 + C_7 C_2) + (C_2 C_1 C_6 + C_2 C_5 C_2 + C_3 C_4 C_2) - (C_2 C_1 C_4 C_2) = 3002$.

%%%%%%%%%%%%%%%%%%%%%%%%%%%%%%%%%%%%%%%%%%%%%%%%%%%%%%%%%%%%%%%%%%
\subsection{Forward incrementation of the leaf labels}
\label{secForwardCyclic}

\noindent In the case that $G$ represents a forward rather than a reverse incrementation of $S$, the roadblocks appear in a triangular region rather than exclusively on the diagonal of the square lattice.

\begin{prop}
If $G$ is obtained by forward incrementation of $S$, then the roadblock set $B_{G,S}$ consists of a triangle of points on and below the diagonal of the square lattice.
\label{propForward}
\end{prop}

\begin{proof}
Consider labels $k_s,k_\ell \in \{1,2,\ldots,n\}$, with $k_s < k_\ell$ and $k_\ell \neq 2$. By definition of forward incrementation, for some component of $G$ with leaves sequentially labeled $k_s, k_s+1, \ldots, k_\ell$ from the cherry toward the root, associated leaves of $S$ are labeled $\pi_{k_s}(G) = k_\ell$, $\pi_{k_s+1}(G)= k_s$, $\pi_{k_s+2}(G)=k_s+1, \ldots, \pi_{k_\ell}(G)=k_\ell-1$.

We use Proposition \ref{propBijection} and compute the minimal internal edge of $S$ ancestral to each leaf $g_k$ of $G$, $k \in \{k_s, k_s+1, \ldots, k_\ell\}$. We obtain $f(g_k) = k_\ell-1$.

The roadblocks are the points $(i,j)$ satisfying $i < f(g_{j+1})$. Hence, the roadblocks are points $(k_s-1,k_s-1), \ldots, (k_\ell-2,k_s-1)$, $(k_s,k_s), \ldots, (k_\ell-2, k_s), \ldots, (k_\ell-2,k_\ell-2)$.
\end{proof}

We can use Catalan's trapezoids to count coalescent histories for forward incrementations, noting that every monotonic path from $(0,0)$ to $(n-1,n-1)$ passes through exactly one point on a diagonal line from the lower right corner of the triangle of roadblocks, $(k_\ell-2,k_s-1)$ for $2 \leqslant k_s \leqslant n-1$ and $(k_\ell-2,1)$ for $k_s=1$, to the bottom edge or right edge of the lattice (Figure \ref{figForward}A).

If $2 \leqslant k_s \leqslant n-1$, then this line has points $(k_\ell-1+c,k_s-2-c)$ for $c=0,1,\ldots,\min(k_s-2,n-k_\ell)$; if $k_s=1$, then the line has a single point $(k_\ell-1,0)$. We can combine the two cases with the Kronecker delta, capturing the line with the expression $(k_\ell-1+c,k_s-2+\delta_{k_s,1}-c)$ for $c=0,1,\ldots,\min(k_s-2+\delta_{k_s,1},n-k_\ell)$.

We can then count monotonic paths from $(0,0)$ to some point on the line and from there to $(n-1,n-1)$.

\begin{prop}
Consider a caterpillar species tree topology $S$ with $n$ leaves and a caterpillar gene tree topology $G$ that differs from $S$ by a forward incrementation described by the component $k_s,\ldots,k_\ell$ of $G$. The number of coalescent histories for $(G,S)$ can be written
\begin{equation*}
\sum_{c=0}^{\min(k_s-2+\delta_{k_s,1},n-k_\ell)} D(k_\ell-1+c,k_s-2+\delta_{k_s,1}-c) D_{k_\ell-k_s-\delta_{k_s,1}+2+2c}
(n-k_\ell-c, n-k_s+1-\delta_{k_s,1}+c).
\end{equation*}
where functions $D$ and $D_m$ follow eqs.~\ref{eqCatalansTriangle} and \ref{eqCatalansTrapezoid}, respectively.
\label{propMultipleForward}
\end{prop}
\begin{proof}
Each monotonic path from $(0,0)$ to $(n-1,n-1)$, proceeds through a point on the diagonal associated with the forward incrementation. The number of paths to arrive at that point from $(0,0)$ is tabulated by Catalan's triangle (eq.~\ref{eqCatalansTriangle}), and the number of paths to reach $(n-1,n-1)$ by Catalan's trapezoids (eq.~\ref{eqCatalansTrapezoid}).
\end{proof}

Figure \ref{figForward}A provides an example. In the figure, $k_s=3$ and $k_\ell=5$, so that each roadblocked monotonic path must pass through $(4,1)$ or $(5,0)$. Figure \ref{figForward}B illustrates the Catalan trapezoid from $(4,1)$ to $(9,9)$, and Figure \ref{figForward}C illustrates the Catalan trapezoid from $(5,0)$ to $(9,9)$. Because the number of paths from $(0,0)$ to $(4,1)$ is 4 and the number of paths from $(0,0)$ to $(5,0)$ is 1, the number of coalescent histories is $4\times 572 + 1 \times 429 = 2717$. This value is returned by the proposition, which gives $\sum_{c=0}^{\min(1,5)} D(4+c,1-c) D_{4+2c}(5-c,8+c)
= D(4,1)D_4(5,8) + D(5,0)D_6(4,9) = 4 \times 572 + 1 \times 429 = 2717$.

Note that we can analyze cases with multiple disjoint forward incrementations by identifying their associated negatively sloping diagonals through which all monotonic paths must pass. The number of coalescent histories can be obtained by a nested sum counting monotonic paths that pass through exactly one point on each diagonal. Changing the perspective to consider the Dyck path associated with the roadblock set for a composition of disjoint forward incrementations, each peak in the Dyck path generates a diagonal, and we can tabulate monotonic paths that pass through points on each of these diagonals.

For example, in Figure \ref{figForward}D, the Dyck path associated with the roadblock set has four peaks. All monotonic paths from $(0,0)$ to $(9,9)$ must pass through two of these, at $(1,0)$ and $(9,8)$. The other two peaks generate diagonals through which all monotonic paths must pass, so that all paths must pass through $(4,1)$ or $(5,0)$ and through $(8,4)$ or $(9,3)$. The number of paths passing through  $(4,1)$ and $(8,4)$ is $D(4,1) D_4(4,3) D_5(1,5) = 700$; the number of paths through $(4,1)$ and $(9,3)$ is $D(4,1) D_4(5,2) D_7(0,6) = 84$; the number through  $(5,0)$ and $(8,4)$ is $D(5,0) D_6(3,4) D_5(1,5) = 175$; and the number through $(5,0)$ and $(9,3)$ is $D(5,0) D_6(4,3) D_7(0,6) = 35$.
In total, the number of paths is 994.

With this perspective, we can see that such an approach to enumeration applies to any Dyck path, not just those that represent disjoint forward incrementations: for every peak in the Dyck path, a diagonal list of points is generated through which each monotonic path from $(0,0)$ to $(n-1,n-1)$ must pass. We consider all possible choices of points, one on each diagonal, and tabulate paths through those points by use of Catalan's triangle and Catalan's trapezoid. For a general pair of caterpillar trees, such an approach can reduce the number of summations in eq.~\ref{eqRecursionCaterpillar2} from $n-1$ to the number of peaks in the associated Dyck path.

The number of Dyck paths of semilength $n$ with exactly $k$ peaks follows the Narayana numbers $N(n,k) = \frac{1}{n} {n \choose k} {n \choose k-1}$ \citep[][Section 6.1]{Deutsch99}. The mean number of peaks in a Dyck path chosen at random then follows $\sum_{k=1}^n \frac{k}{n} {n \choose k} {n \choose k-1} / C_n$, which, by noting $k {n \choose k} = n {n-1 \choose k-1}$ and applying eq.~5.23 in Table 169 of \cite{GrahamEtAl94} to complete the summation, gives a mean of $(n+1)/2$. Thus, because we consider semi-length $n-1$, this approach reduces the mean number of nested summations from $n-1$ in eq.~\ref{eqRecursionCaterpillar2}  to $n/2$.

%%%%%%%%%%%%%%%%%%%%%%%%%%%%%%%%%%%%%%%%%%%%%%%%%%%%%%%%%%%%%%%%%%
%%%%%%%%%%%%%%%%%%%%%%%%%%% Figure 7 %%%%%%%%%%%%%%%%%%%%%%%%%%%%%
%%%%%%%%%%%%%%%%%%%%%%%%%%%%%%%%%%%%%%%%%%%%%%%%%%%%%%%%%%%%%%%%%%
\begin{figure}[tpb]
\vskip -1cm
\begin{center}
\includegraphics[width=7in]{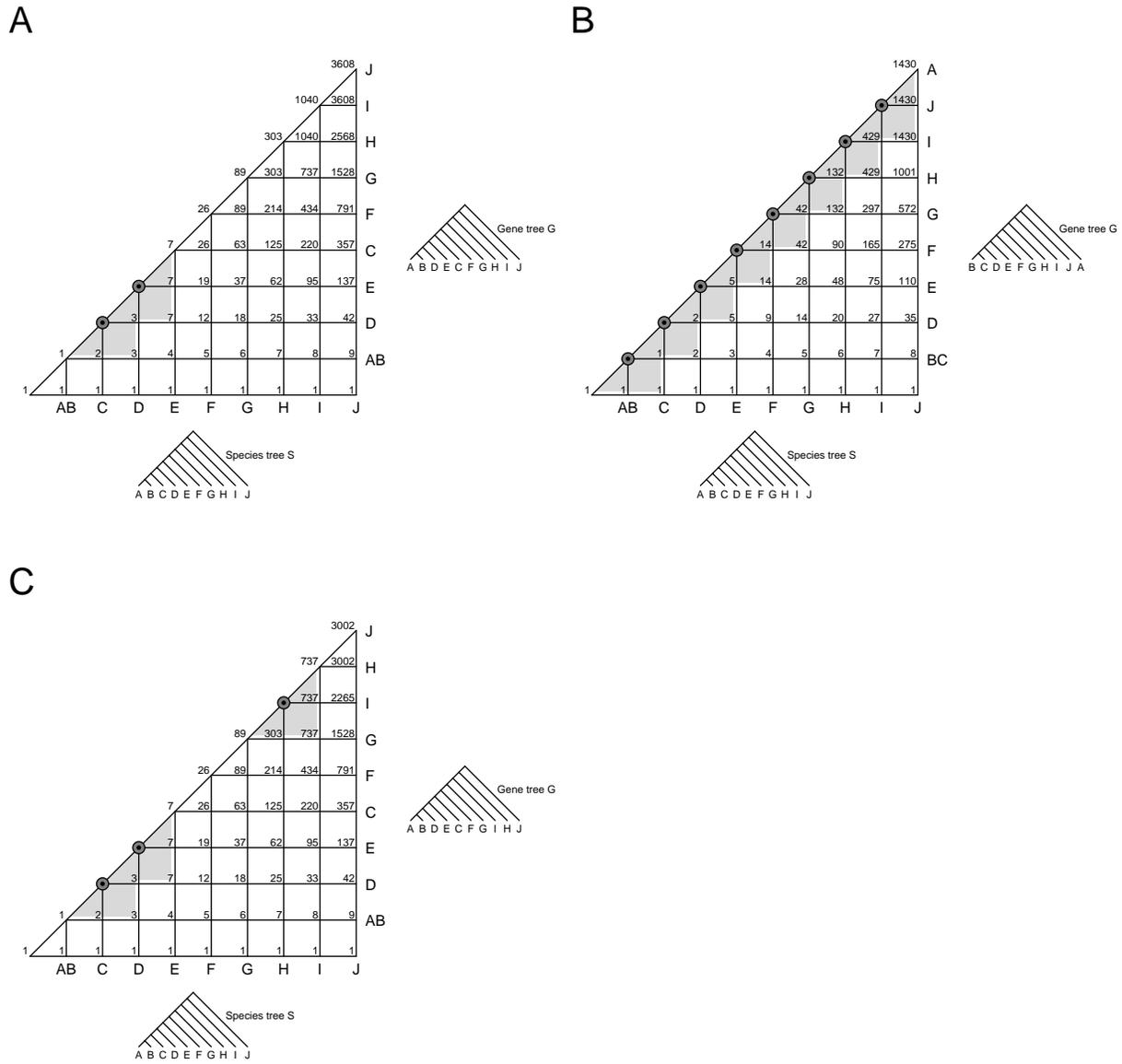}
\end{center}
\vskip -6.5cm
\caption{The number of coalescent histories for reverse incrementations. (A) $G$ differs from $S$ by a reverse incrementation. (B) $G$ differs from $S$ by a reverse incrementation that includes all labels. (C) $G$ differs from $S$ by a composition of multiple disjoint reverse incrementations.}
\label{figReverse}
\end{figure}
%%%%%%%%%%%%%%%%%%%%%%%%%%%%%%%%%%%%%%%%%%%%%%%%%%%%%%%%%%%%%%%%%%
%%%%%%%%%%%%%%%%%%%%%%%%%%%%%%%%%%%%%%%%%%%%%%%%%%%%%%%%%%%%%%%%%%
%%%%%%%%%%%%%%%%%%%%%%%%%%%%%%%%%%%%%%%%%%%%%%%%%%%%%%%%%%%%%%%%%%

%%%%%%%%%%%%%%%%%%%%%%%%%%%%%%%%%%%%%%%%%%%%%%%%%%%%%%%%%%%%%%%%%%
%%%%%%%%%%%%%%%%%%%%%%%%%%% Figure 8 %%%%%%%%%%%%%%%%%%%%%%%%%%%%%
%%%%%%%%%%%%%%%%%%%%%%%%%%%%%%%%%%%%%%%%%%%%%%%%%%%%%%%%%%%%%%%%%%
\begin{figure}[tpb]
\vskip -5cm
\begin{center}
\includegraphics[width=7in]{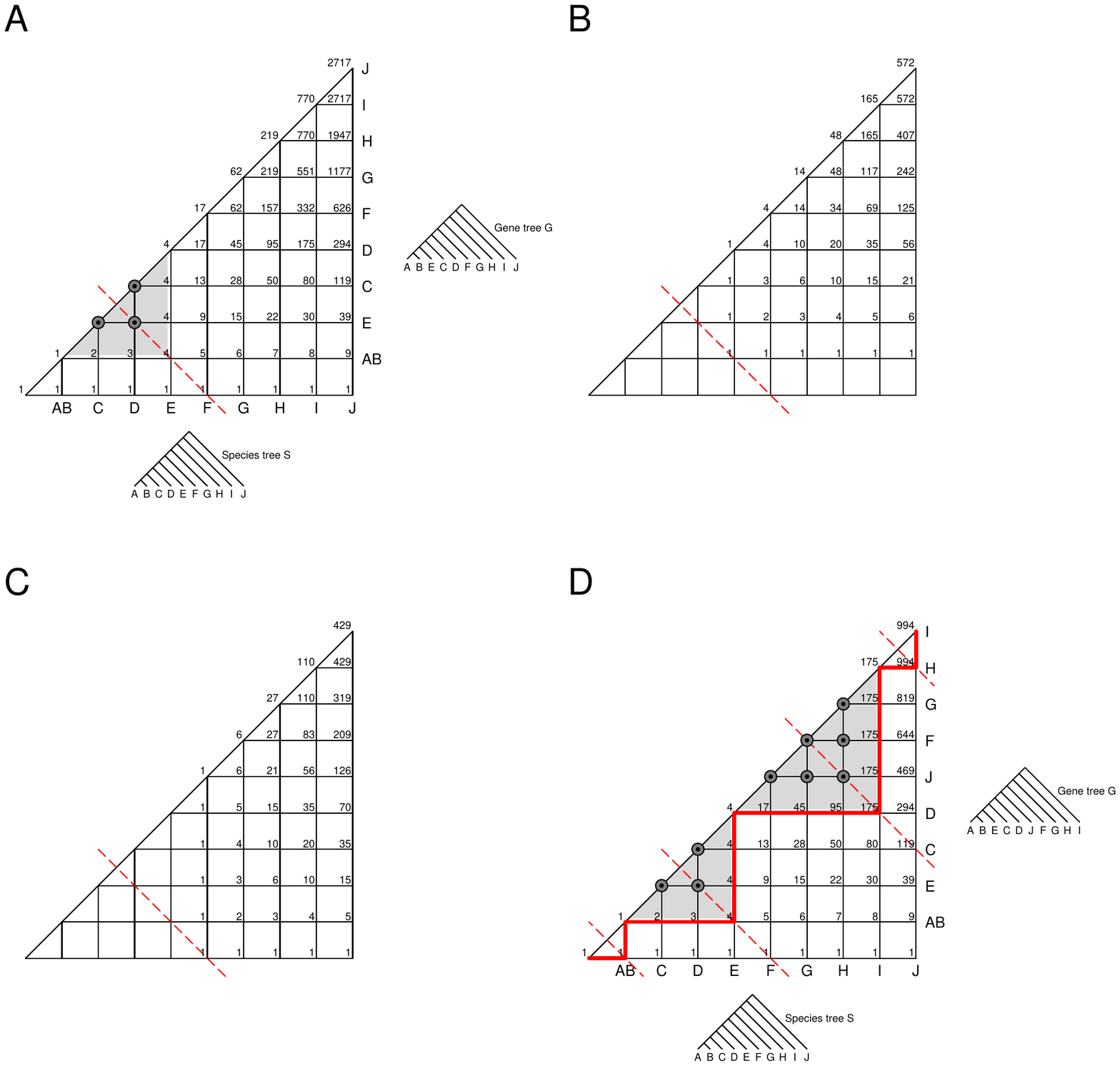}
\end{center}
\vskip -6.5cm
\caption{The number of coalescent histories for forward incrementations. (A) $G$ differs from $S$ by a forward incrementation. All paths must pass through the dashed red line. (B) The number of paths from $(4,1)$ on the dashed red line to $(9,9)$. (C) The number of paths from $(5,0)$ on the dashed red line to $(9,9)$. (D) $G$ differs from $S$ by a composition of two forward incrementations. All paths must pass through the four dashed red lines. The solid line represents the Dyck path associated with the roadblock set (see Figure \ref{figRoadblockSetSymmetric}).}
\label{figForward}
\end{figure}
%%%%%%%%%%%%%%%%%%%%%%%%%%%%%%%%%%%%%%%%%%%%%%%%%%%%%%%%%%%%%%%%%%
%%%%%%%%%%%%%%%%%%%%%%%%%%%%%%%%%%%%%%%%%%%%%%%%%%%%%%%%%%%%%%%%%%
%%%%%%%%%%%%%%%%%%%%%%%%%%%%%%%%%%%%%%%%%%%%%%%%%%%%%%%%%%%%%%%%%%

%%%%%%%%%%%%%%%%%%%%%%%%%%%%%%%%%%%%%%%%%%%%%%%%%%%%%%%%%%%%%%%%%%
\section{Discussion}
\label{secDiscussion}

\noindent We have studied coalescent histories for non-matching caterpillar gene trees and species trees, showing that as in the matching case, the number of coalescent histories for non-matching caterpillars can be computed using monotonic paths that do not cross the diagonal of a square lattice (Section \ref{secBijection}). The recursion for the number of coalescent histories that applies for arbitrary gene trees and species trees simplifies for non-matching caterpillars to a non-recursive formula dependent only on the caterpillar topologies (Section \ref{secNonrecursive}). Using these results, we have counted  coalescent histories for non-matching caterpillars differing by nearest-neighbor-interchange (Section \ref{secNNI}). By studying reverse and forward incrementation, we have also counted coalescent histories for caterpillars differing by subtree-prune-and-regraft (Sections \ref{secReverseCyclic} and \ref{secForwardCyclic}).

The bijection that connects coalescent histories and monotonic paths (Proposition \ref{propBijection}) makes use of \emph{roadblocks}, lattice points through which paths are not permitted to travel. Roadblocks occur such that if a point $(i,j)$ is a roadblock for $i \geqslant j$, then $(i,k)$ is also a roadblock for each $k$ with $j \leqslant k \leqslant i$, as is $(\ell,j)$ for each $\ell$ with $j \leqslant \ell \leqslant i$ (Remark \ref{remBGS}). Enumeration of roadblocked monotonic paths given a roadblock set connects to Catalan's triangle and trapezoids, enabling enumeration of the associated coalescent histories. Interestingly, the distinct roadblock sets can themselves be put into bijection with the monotonic paths that do not cross the diagonal of a square lattice, so that their number also follows the Catalan sequence (Section \ref{secRoadblockSets}).

Our construction linking coalescent histories and roadblocked monotonic paths enables a simple proof of a result of \cite{DegnanAndRhodes15} that for a fixed number of leaves, matching caterpillar trees have more coalescent histories than do non-matching caterpillar trees (Corollary \ref{coroStrictlyGreater}). In particular, the lattice construction that enumerates coalescent histories for a non-matching pair of caterpillar trees contains at least one roadblock, whereas the lattice for matching caterpillars has no roadblocks and therefore has more monotonic paths. For a fixed caterpillar species tree, we have identified exactly which non-matching caterpillar gene tree generates the most coalescent histories: it is immediate that this gene tree differs from the species tree by a single NNI move, as the caterpillars differing from the species tree by one NNI move are the only ones that produce only one roadblock. We find that the specific NNI move affecting leaves nearest the ``middle'' of the species tree generates the largest number of coalescent histories, and that as the number of leaves increases, this value is asymptotically equivalent to the Catalan number $C_{n-1}$ (Section \ref{secNNI}).

The case in which the gene tree differs from the species tree by reverse incrementation produces an elegant result. Recalling that the number of coalescent histories for matching caterpillars is described by a Catalan number, if the gene tree is obtained by a reverse incrementation affecting all leaf labels of the species tree, then the number of coalescent histories is given by the next-smaller Catalan number (Section \ref{secReverseCyclic}). The case of forward incrementation is more complex, but it can be analyzed using Catalan's trapezoids and suggests further connections to the analysis of Dyck paths (Section \ref{secForwardCyclic}).

This study provides some of the first systematic closed-form results concerning coalescent histories for non-matching gene trees and species trees. Our approach applies only to caterpillars, however, as the bijection with roadblocked monotonic paths relies on the fact that the internal nodes of a caterpillar tree can be placed in a sequence such that all pairs of internal nodes have an ancestor-descendant relationship. It does suggest, however, that connections to other combinatorial structures such as Dyck paths can assist in enumerating coalescent histories for non-matching gene trees and species trees beyond use of the recursion in eq.~\ref{eqRecursion}.

A question that remains open is that the set of integers that could equal the number of coalescent histories for some caterpillar gene tree and species tree remains unknown. Rosenberg \& Degnan (2010, Table 1) \nocite{RosenbergAndDegnan10} observed that for fixed species trees $S$ of size $n$ and certain values $t$, particularly small ones, large numbers of pairs $(G,S)$ had exactly $t$ coalescent histories, and Rosenberg (2019) \nocite{Rosenberg19} enumerated the pairs $(G,S)$ with exactly 1 coalescent history (the \emph{lonely pairs}). We and \cite{DegnanAndRhodes15} have shown that if $G$ and $S$ are caterpillars, then only values $t \leqslant C_{n-1}$ can represent the number of coalescent histories. Our NNI results show that all values in the open interval $(C_{n-1}-C_{\lfloor (n-1)/2 \rfloor} C_{\lceil (n-1)/2 \rceil}, C_{n-1})$ cannot be the number of coalescent histories for $(G,S)$. For fixed caterpillar $S$ with $n$ leaves, it is useful to obtain the size of the set of values of $t$ for which the pair $(G,S)$ has exactly $t$ coalescent histories. We observed that $(C_{n-1} + {n-1 \choose \lfloor (n-1)/2 \rfloor})/2$ provides an upper bound (Section \ref{secRoadblockSets}).

We note that the question of identifying the integers that represent the number of coalescent histories for some $(G,S)$ with $n$ leaves can be phrased entirely in terms of roadblocked monotonic paths without reference to coalescent histories. Describe a lattice as \emph{monotonically roadblocked} if for each roadblock $(i,j)$ with $i \geqslant j$, $(i,k)$ is also a roadblock for each $k$ with $j \leqslant k \leqslant i$, and $(\ell,j)$ is a roadblock for each $\ell$ with $j \leqslant \ell \leqslant i$. We seek the number of integers that represent the number of monotonic paths that do not cross the diagonal of some monotonically roadblocked lattice. That the bijection between coalescent histories and roadblocked monotonic paths raises such questions illustrates that constructions enabled by this bijection can be fruitful for studies of the properties of the paths themselves.

%%%%%%%%%%%%%%%%%%%%%%%%%%%%%%%%%%%%%%%%%%%%%%%%%%%%%%%%%%%%%%%%%%%%%%%
%%%%%%%%%%%%%%%%%%%%%%%%%%%%%%%%%%%%%%%%%%%%%%%%%%%%%%%%%%%%%%%%%%%%%%%
%%%%%%%%%%%%%%%%%%%%%%%%%%%%%%%%%%%%%%%%%%%%%%%%%%%%%%%%%%%%%%%%%%%%%%%
\vskip .3cm
\noindent
{%\small
{\bf Acknowledgments.} We thank E.~Allman, J.~Degnan, F.~Disanto, and J.~Rhodes for helpful discussions. We acknowledge NIH grants R01 GM117590 and R01 GM131404 for support.
}
{%\footnotesize
\bibliography{map3}}

\begin{thebibliography}{}

\bibitem[Bonin {\em et~al.}, 2003]{BoninEtAl03}
Bonin, J., {de Mier}, A., {and} Noy, M. 2003.
\newblock Lattice path matroids: enumerative aspects and {Tutte} polynomials,
  {\em J. Comb. Theory Ser. A} {\bf 104}, 63--94.

\bibitem[Degnan, 2005]{Degnan05}
Degnan, J.~H. 2005.
\newblock ``Gene tree distributions under the coalescent process''.
\newblock PhD thesis University of New Mexico Albuquerque.

\bibitem[Degnan \& Rhodes, 2015]{DegnanAndRhodes15}
Degnan, J.~H. {and} Rhodes, J.~A. 2015.
\newblock There are no caterpillars in a wicked forest, {\em Theor. Pop. Biol.}
  {\bf 105}, 17--23.

\bibitem[Degnan \& Rosenberg, 2009]{DegnanAndRosenberg09}
Degnan, J.~H. {and} Rosenberg, N.~A. 2009.
\newblock Gene tree discordance, phylogenetic inference and the multispecies
  coalescent, {\em Trends Ecol. Evol.} {\bf 24}, 332--340.

\bibitem[Degnan {\em et~al.}, 2012]{DegnanEtAl12:mathbiosci}
Degnan, J.~H., Rosenberg, N.~A., {and} Stadler, T. 2012.
\newblock The probability distribution of ranked gene trees on a species tree,
  {\em Math. Biosci.} {\bf 235}, 45--55.

\bibitem[Degnan \& Salter, 2005]{DegnanAndSalter05}
Degnan, J.~H. {and} Salter, L.~A. 2005.
\newblock Gene tree distributions under the coalescent process, {\em Evolution}
  {\bf 59}, 24--37.

\bibitem[Deng {\em et~al.}, 2015]{DengEtAl15}
Deng, L.-H., Deng, Y.-P., {and} Shapiro, L.~W. 2015.
\newblock The {Riordan} group and symmetric lattice paths, {\em J. Shandong
  Univ.} {\bf 50}, 82--89.

\bibitem[Deutsch, 1999]{Deutsch99}
Deutsch, E. 1999.
\newblock Dyck path enumeration, {\em Discr. Math.} {\bf 204}, 167--202.

\bibitem[Disanto \& Rosenberg, 2015]{DisantoAndRosenberg15}
Disanto, F. {and} Rosenberg, N.~A. 2015.
\newblock Coalescent histories for lodgepole species trees, {\em J. Comput.
  Biol.} {\bf 22}, 918--929.

\bibitem[Disanto \& Rosenberg, 2016]{DisantoAndRosenberg16}
Disanto, F. {and} Rosenberg, N.~A. 2016.
\newblock Asymptotic properties of the number of matching coalescent histories
  for caterpillar-like families of species trees, {\em IEEE/ACM Trans. Comput.
  Biol. Bioinf.} {\bf 13}, 913--925.

\bibitem[Graham {\em et~al.}, 2008]{GrahamEtAl94}
Graham, R.~L., Knuth, D.~E., {and} Patashnik, O. 2008.
\newblock ``Concrete Mathematics'', Addison-Wesley, Boston, 2nd edition.

\bibitem[Maddison, 1997]{Maddison97}
Maddison, W.~P. 1997.
\newblock Gene trees in species trees, {\em Syst. Biol.} {\bf 46}, 523--536.

\bibitem[Pamilo \& Nei, 1988]{PamiloAndNei88}
Pamilo, P. {and} Nei, M. 1988.
\newblock Relationships between gene trees and species trees, {\em Mol. Biol.
  Evol.} {\bf 5}, 568--583.

\bibitem[Reuveni, 2014]{Reuveni14}
Reuveni, S. 2014.
\newblock Catalan's trapezoids, {\em Prob. Eng. Inform. Sci.} {\bf 28},
  353--361.

\bibitem[Rosenberg, 2002]{Rosenberg02}
Rosenberg, N.~A. 2002.
\newblock The probability of topological concordance of gene trees and species
  trees, {\em Theor. Pop. Biol.} {\bf 61}, 225--247.

\bibitem[Rosenberg, 2007]{Rosenberg07:jcb}
Rosenberg, N.~A. 2007.
\newblock Counting coalescent histories, {\em J. Comput. Biol.} {\bf 14},
  360--377.

\bibitem[Rosenberg, 2013]{Rosenberg13:tcbb}
Rosenberg, N.~A. 2013.
\newblock Coalescent histories for caterpillar-like families, {\em IEEE/ACM
  Trans. Comp. Biol. Bioinf.} {\bf 10}, 1253--1262.

\bibitem[Rosenberg, 2019]{Rosenberg19}
Rosenberg, N.~A. 2019.
\newblock Enumeration of lonely pairs of gene trees and species trees by means
  of antipodal cherries, {\em Adv. Appl. Math.} {\bf 102}, 1--17.

\bibitem[Rosenberg \& Degnan, 2010]{RosenbergAndDegnan10}
Rosenberg, N.~A. {and} Degnan, J.~H. 2010.
\newblock Coalescent histories for discordant gene trees and species trees,
  {\em Theor. Pop. Biol.} {\bf 77}, 145--151.

\bibitem[Rosenberg \& Tao, 2008]{RosenbergAndTao08}
Rosenberg, N.~A. {and} Tao, R. 2008.
\newblock Discordance of species trees with their most likely gene trees: the
  case of five taxa, {\em Syst. Biol.} {\bf 57}, 131--140.

\bibitem[Stadler \& Degnan, 2012]{StadlerAndDegnan12}
Stadler, T. {and} Degnan, J.~H. 2012.
\newblock A polynomial time algorithm for calculating the probability of a
  ranked gene tree given a species tree, {\em Alg. Mol. Biol.} {\bf 7}, 7.

\bibitem[Stanley, 1999]{Stanley99}
Stanley, R.~P. 1999.
\newblock ``Enumerative Combinatorics Volume 2'', Cambridge University Press,
  New York.

\bibitem[Stanley, 2015]{Stanley15}
Stanley, R.~P. 2015.
\newblock ``An Introduction to Probability Theory'', Cambridge University
  Press, Cambridge.

\bibitem[Steel, 2016]{Steel16}
Steel, M. 2016.
\newblock ``Phylogeny: Discrete and Random Processes in Evolution'', Society
  for Industrial and Applied Mathematics, Philadelphia.

\bibitem[Than \& Nakhleh, 2009]{ThanAndNakhleh09}
Than, C. {and} Nakhleh, L. 2009.
\newblock Species tree inference by minimizing deep coalescences, {\em PLoS
  Comp. Biol.} {\bf 5}, e1000501.

\bibitem[Than {\em et~al.}, 2007]{ThanEtAl07}
Than, C., Ruths, D., Innan, H., {and} Nakhleh, L. 2007.
\newblock Confounding factors in {HGT} detection: statistical error, coalescent
  effects, and multiple solutions, {\em J. Comput. Biol.} {\bf 14}, 517--535.

\bibitem[Wu, 2012]{Wu12}
Wu, Y. 2012.
\newblock Coalescent-based species tree inference from gene tree topologies
  under incomplete lineage sorting by maximum likelihood, {\em Evolution} {\bf
  66}, 763--775.

\bibitem[Wu, 2016]{Wu16}
Wu, Y. 2016.
\newblock An algorithm for computing the gene tree probability under the
  multispecies coalescent and its application in the inference of population
  tree, {\em Bioinformatics} {\bf 32}, i225--i233.

\end{thebibliography}
\end{document}